\documentclass[12pt,letterpaper, onecolumn]{IEEEtran}
\usepackage[margin=1in]{geometry}
\usepackage{setspace}

\newif\ifthreevalue
\threevaluetrue

\usepackage[utf8]{inputenc} 
\usepackage[T1]{fontenc}
\usepackage{url}  

\usepackage[cmex10]{amsmath}  
\usepackage{mleftright}       
\mleftright                   

\usepackage{graphicx}         
\usepackage{booktabs}         

\usepackage{amsmath,amssymb,amsfonts,amsthm}

\let\emptyset\varnothing
\usepackage{enumitem}
\usepackage{mathtools}
\usepackage{algorithmic}
\usepackage{textcomp}
\usepackage[table,xcdraw]{xcolor}
\usepackage{tikz}

\usepackage{subcaption}
\usepackage{multicol,multirow}
\usepackage{scalerel}

\usepackage[style=ieee,citestyle=numeric-comp,sorting=nty]{biblatex}
\newtheorem{theorem}{Theorem}
\newtheorem{remark}{Remark}

\newtheorem{example}{Example}
\newtheorem{lemma}{Lemma}
\newtheorem{definition}{Definition}

\newcommand{\blue}[1]{\textcolor{blue}{#1}}

\newcommand{\bC}{\mathbb{C}}

\newcommand{\bF}{\mathbb{F}}

\newcommand{\bN}{\mathbb{N}}

\newcommand{\bR}{\mathbb{R}}

\newcommand{\cA}{\mathcal{A}}

\newcommand{\cF}{\mathcal{F}}

\newcommand{\cS}{\mathcal{S}}

\newcommand{\cU}{\mathcal{U}}

\newcommand{\boldb}{\mathbf{b}}

\newcommand{\boldd}{\mathbf{d}}
\newcommand{\bolde}{\mathbf{e}}

\newcommand{\boldq}{\mathbf{q}}

\newcommand{\boldu}{\mathbf{u}}
\newcommand{\boldv}{\mathbf{v}}
\newcommand{\boldw}{\mathbf{w}}
\newcommand{\boldx}{\mathbf{x}}
\newcommand{\boldy}{\mathbf{y}}
\newcommand{\boldz}{\mathbf{z}}

\newcommand{\boldU}{\mathbf{U}}

\newcommand{\boldW}{\mathbf{W}}

\DeclareMathOperator{\FtwoSpan}{\operatorname{Span}_{\bF_2}}
\DeclareMathOperator{\RSpan}{\operatorname{Span}_{\bR}}
\DeclareMathOperator{\diag}{\operatorname{diag}}
\DeclareMathOperator*{\addsets}{\scalerel*{+}{\sum}}

\DeclareSymbolFont{bbold}{U}{bbold}{m}{n}
\DeclareSymbolFontAlphabet{\mathbbold}{bbold}
\newcommand{\1}{\mathbbold{1}}

\author{\textbf{Zirui (Ken) Deng}, \textbf{Vinayak Ramkumar}, \textbf{Rawad Bitar}, and \textbf{Netanel Raviv}
 \thanks{This paper was presented in part at the IEEE International Symposium on Information Theory (ISIT), 2022. 
 
 Zirui Deng and Netanel Raviv are with the Department of Computer Science and Engineering at Washington University in St. Louis, MO, USA. Emails: d.ken@wustl.edu, netanel.raviv@wustl.edu.
 
 Vinayak Ramkumar is jointly with the Department of Electrical Engineering-Systems at Tel Aviv University, Israel and the Department of Computer Science and Engineering at Washington University in St. Louis, MO, USA. Email: vinram93@gmail.com. 
 
 Rawad Bitar is with the Institute for Communications Engineering at Technical University of Munich, Germany. Email: rawad.bitar@tum.de. 
 
 }}
\begin{document}

\title{Private Inference in Quantized Models}

\maketitle

\pagestyle{plain}
\begin{abstract}
    A typical setup in many machine learning scenarios involves a server that holds a model and a user that possesses data, and the challenge is to perform inference while safeguarding the privacy of both parties. \textit{Private Inference} has been extensively explored in recent years, mainly from a cryptographic standpoint via techniques like homomorphic encryption and multiparty computation. These approaches often come with high computational overhead and may degrade the accuracy of the model. In our work, we take a different approach inspired by the \textit{Private Information Retrieval} literature. We view private inference as the task of retrieving inner products of parameter vectors with the data, a fundamental operation in many machine learning models. We introduce schemes that enable such retrieval of inner products for models with~\textit{quantized}  (i.e., restricted to a finite set) weights; such models are extensively used in practice due to a wide range of benefits. In addition, our schemes uncover a fundamental tradeoff between user and server privacy. Our information-theoretic approach is applicable to a wide range of problems and robust in privacy guarantees for both the user and the server. 
\end{abstract}

\thispagestyle{plain}


\section{Introduction}
Loss of privacy in the information age raises various concerns ranging from financial security issues to mental health problems. At the heart of this lies a dilemma for individual users: they must decide between relinquishing their personal data to service providers and forfeiting the capacity to engage in fundamental day-to-day tasks. On the other hand, training machine learning models is a laborious task which requires ample knowledge and resources. As such, the models themselves are the intellectual property of their owners, and should be kept private as well.

Information privacy becomes particularly important in the context of the \textit{inference} phase within the machine learning pipeline. In this scenario, a server (e.g., a service provider) possesses an already-trained model (e.g., a neural network, a logistic/linear regression model, a
linear classifier, etc.) and offers a user access to this model in exchange for a fee. To facilitate the inference process, the user engages in an information exchange with the server, allowing the model to make predictions or perform computations on their input data. A \textit{private inference} (PI) protocol is one that ensures a certain level of privacy for both parties involved in this exchange. 

PI has been the subject of extensive research in recent years, primarily from a cryptographic perspective through primitives like homomorphic encryption \cite{gilad2016cryptonets, hesamifard2017cryptodl,reagen2021cheetah} and multiparty computing \cite{mohassel2017secureml, riazi2019XONN}. However, each of these techniques has its limitations, including high computational overhead for non-polynomial functions, degraded accuracy, the need for multiple communication rounds, among others; refer to \cite{boemer2020mp2ml} for a comprehensive overview. In this paper, we take a perpendicular direction and develop an information-theoretic
approach to private inference. Among the benefits of the latter over the former are resilience against computationally unbounded adversaries, clear privacy guarantees and simplicity of computation.

Our approach begins with the observation that the computation involved in most machine learning models starts with the extraction of one or more \textit{signals} from the data, i.e., values of the form $\boldw\boldx^\intercal$, where $\boldw$ represents a weight vector associated with the model and $\boldx$ represents a vector of features, both of them usually real-valued. This is the case in linear classification (sign($\boldw\boldx^\intercal$)), linear regression ($\boldw\boldx^\intercal$), logistic regression ($1/(1+\exp(-\boldw\boldx^\intercal))$), and many others. In order to provide clear information-theoretic guarantees, especially from the server's perspective, we leverage the notion of \textit{quantized} models. It has been demonstrated that by constraining the weight values of $\boldw$ to $\{\pm 1\}$, for instance, one can obtain significant gains in terms of simplicity of implementation \cite{mcdanel2017embedded}, as well as resilience against adversarial perturbations \cite{galloway2018attacking, raviv2020codnnrobust, raviv2021enhancing}, while keeping the models efficiently trainable \cite{hubara2016binarized, teerapittayanon2017distributed}. From the user's standpoint, we treat the data distribution as continuous, and our approach offers privacy guarantees with respect to the dimension in which the server obtains knowledge about the data. See Section~\ref{section:justificationforell} for a justification of this user-privacy metric based on the Infomax principle~\cite{linsker1988self}.

The method we follow in this paper draws inspiration from the vast \textit{private information retrieval} (PIR) literature \cite{chor1998pir,yekhanin2010pir,beimel2005general,kushilevitz1997replication,ostrovsky2007survey}. In PIR, a user wishes to retrieve an entry from a distributed dataset while keeping the identity of that entry private from potentially colluding servers (note the inverted role of “user” and “server” here with respect to our work). Retrieving an entry $x_i$ from a dataset $\boldx = (x_1, \ldots, x_n)$ while keeping the index $i$ private can be regarded as retrieving $\boldx\bolde_i^\intercal$, where $\bolde_i$ is the $i$-th unit vector that must remain private. A natural generalization is retrieving $\boldx\boldw^\intercal$, rather than $\boldx\bolde_i^\intercal$, for some weight vector $\boldw$ that must remain private, bringing us to the focus of this work. In this generalization—often referred to as~\textit{private computation} \cite{raviv2020privatepolynomial, sun2018privatecomputation}—the inner product $\boldw\boldx^\intercal$ is computed over a finite field, unlike private inference. Another discernible difference is that PIR almost exclusively discusses one user and multiple servers, among which collusion is restricted in some way. This is usually not the case in private inference settings, as both the
server and the user are cohesive entities.

This paper is structured as follows. Preliminaries and general problem setup are given in Section \ref{section:preandsetup}. Our protocols for the PI problem where the weights are restricted to $\{\pm 1\}$ are given in Section~\ref{section:pmone}. 
In Section \ref{section:jointretrieval}, we seek to address the more general problem where the server holds multiple $\{\pm 1\}$-vectors and aims to retrieve multiple inner products with the data vector at the same time; this is motivated by models which require such preliminary computation, such as neural networks. In Section \ref{section:perfectandhard}, we consider the case where the weight vector takes values from certain finite sets of interest and show an intriguing connection to an additive combinatorial notion called \textit{coefficient complexity}; this notion was recently studied for applications in distributed computations. In Section~\ref{section:optimality}, we derive a general tradeoff between the privacy of the server and that of the user, as well as a lower bound for the communication cost of the protocol. Then, we analyze the optimality of various protocols presented in earlier sections in terms of the privacy tradeoff and the bound on communication cost. Finally, in Section \ref{section:conclusion}, we propose directions for possible future work.

\section{Preliminaries}\label{section:preandsetup}

\subsection{Problem statement}
We consider PI problems in which  entries of weight vectors are restricted to a finite set. To be specific, a server holds a weight vector $\boldw \in A^n$ for some finite set $A$, where $\boldw$ is randomly chosen from a uniform distribution $W = \textrm{Unif}(A^n)$, and a user holds a data vector $\boldx \in \bR^n$, randomly chosen from an unknown continuous data distribution $X$. Modeling the weight vector as a uniform random variable reflects the user's lack of knowledge about the weights, whereas modeling the data as taken from an unknown distribution is a common practice in machine learning. The server aims to retrieve the inner product~$\boldw\boldx^\intercal$ (computed over $\bR$) while ensuring some degree of privacy for the server and the user. Privacy metrics for both parties will be defined shortly. 

Following the computation of $\boldw\boldx^\intercal$, the server feeds the result into some pre-trained model~$f$ to obtain the inference~$f(\boldw\boldx^\intercal)$, which is then sent back to the user. In what follows we focus on the retrieval of~$\boldw\boldx^\intercal$ and its associated privacy; the remaining part of the inference (such as inner layers in a neural network) remains perfectly private. Note that due to the disclosure of $f(\boldw\boldx^\intercal)$ to the user, some privacy loss is unavoidable, regardless of the privacy guarantees. For instance, the user may repeatedly query multiple vectors $\{\boldx_i\}_{i=1}^{N}$, collect
their inferences $\{f(\boldw\boldx_i^\intercal)\}_{i=1}^N$ by following the protocol $N$ times, and then train a model similar to $f$.

Inspired by the PIR literature, we focus on protocols of the query-answer form as follows:

\begin{enumerate}[label=\arabic*)]
    \item The server sends to the user a query vector $\boldq \sim Q$; the distribution $Q$ from which $\boldq$ is chosen is a deterministic function of $W$, i.e., $H(Q\vert W) = 0$, where $H$ is the entropy function.
    \item The user computes some $\ell$ vectors $\boldv_1, \ldots, \boldv_\ell \in \bR^n$ deterministically from $\boldq$, and sends an answer $\cA = \{\boldv_i\boldx^\intercal\}_{i=1}^{\ell }$ to the server.
    \item The server combines the elements of $\cA$ to retrieve the value $\boldw\boldx^\intercal$. 
\end{enumerate}

We note that the structure of the problem allows the server to send the same query $\boldq$ to multiple users without any loss of privacy. This way, the server may post $\boldq$ in some public forum (e.g., the server’s website) and save all future communication with users interested
in inference. These two interpretations of the problem are illustrated in Fig.~\ref{fig:system}. The upcoming information-theoretic analysis guarantees that $\boldq$ can be publicly available indefinitely, while the privacy of $\boldw$ remains protected against any computational power that exists or may exist in the future (up to the inevitable learning attack mentioned earlier). 

\tikzset{every picture/.style={line width=0.5pt}} 
\begin{figure}
\centering 
\begin{tikzpicture}[x=0.5pt,y=0.5pt,yscale=-1,xscale=1]

\draw   (28,71) -- (198,71) -- (198,241) -- (28,241) -- cycle ;
\draw   (376,72) -- (616,72) -- (616,239) -- (376,239) -- cycle ;
\draw    (375,109) -- (201,108.01) ;
\draw [shift={(198,108)}, rotate = 0.32] [color={rgb, 255:red, 0; green, 0; blue, 0 }  ][line width=0.75]    (10.93,-3.29) .. controls (6.95,-1.4) and (3.31,-0.3) .. (0,0) .. controls (3.31,0.3) and (6.95,1.4) .. (10.93,3.29)   ;
\draw    (199,154) -- (372,154.99) ;
\draw [shift={(375,155)}, rotate = 180.33] [color={rgb, 255:red, 0; green, 0; blue, 0 }  ][line width=0.75]    (10.93,-3.29) .. controls (6.95,-1.4) and (3.31,-0.3) .. (0,0) .. controls (3.31,0.3) and (6.95,1.4) .. (10.93,3.29)   ;
\draw    (375,201) -- (201,200.01) ;
\draw [shift={(198,200)}, rotate = 0.32] [color={rgb, 255:red, 0; green, 0; blue, 0 }  ][line width=0.75]    (10.93,-3.29) .. controls (6.95,-1.4) and (3.31,-0.3) .. (0,0) .. controls (3.31,0.3) and (6.95,1.4) .. (10.93,3.29)   ;

\draw   (27,346) -- (217,346) -- (217,409) -- (27,409) -- cycle ;
\draw   (375,441) -- (616,441) -- (616,732) -- (375,732) -- cycle ;
\draw    (407,440) -- (407,376) -- (219,376) ;
\draw [shift={(217,376)}, rotate = 360] [color={rgb, 255:red, 0; green, 0; blue, 0 }  ][line width=0.75]    (10.93,-3.29) .. controls (6.95,-1.4) and (3.31,-0.3) .. (0,0) .. controls (3.31,0.3) and (6.95,1.4) .. (10.93,3.29)   ;
\draw   (125,441) -- (249,441) -- (249,512) -- (125,512) -- cycle ;
\draw   (125,661) -- (249,661) -- (249,732) -- (125,732) -- cycle ;
\draw   (125,549) -- (249,549) -- (249,620) -- (125,620) -- cycle ;
\draw    (248,463) -- (373,463) ;
\draw [shift={(375,463)}, rotate = 180] [color={rgb, 255:red, 0; green, 0; blue, 0 }  ][line width=0.75]    (10.93,-3.29) .. controls (6.95,-1.4) and (3.31,-0.3) .. (0,0) .. controls (3.31,0.3) and (6.95,1.4) .. (10.93,3.29)   ;
\draw    (374,490) -- (252,490) ;
\draw [shift={(250,490)}, rotate = 360] [color={rgb, 255:red, 0; green, 0; blue, 0 }  ][line width=0.75]    (10.93,-3.29) .. controls (6.95,-1.4) and (3.31,-0.3) .. (0,0) .. controls (3.31,0.3) and (6.95,1.4) .. (10.93,3.29)   ;
\draw    (374,598) -- (366,598) -- (356,598) -- (252,598) ;
\draw [shift={(250,598)}, rotate = 360] [color={rgb, 255:red, 0; green, 0; blue, 0 }  ][line width=0.75]    (10.93,-3.29) .. controls (6.95,-1.4) and (3.31,-0.3) .. (0,0) .. controls (3.31,0.3) and (6.95,1.4) .. (10.93,3.29)   ;
\draw    (249,571) -- (374,571) ;
\draw [shift={(376,571)}, rotate = 180] [color={rgb, 255:red, 0; green, 0; blue, 0 }  ][line width=0.75]    (10.93,-3.29) .. controls (6.95,-1.4) and (3.31,-0.3) .. (0,0) .. controls (3.31,0.3) and (6.95,1.4) .. (10.93,3.29)   ;
\draw    (248,683) -- (373,683) ;
\draw [shift={(375,683)}, rotate = 180] [color={rgb, 255:red, 0; green, 0; blue, 0 }  ][line width=0.75]    (10.93,-3.29) .. controls (6.95,-1.4) and (3.31,-0.3) .. (0,0) .. controls (3.31,0.3) and (6.95,1.4) .. (10.93,3.29)   ;
\draw    (374,709) -- (366,709) -- (356,709) -- (349,709) -- (252,709) ;
\draw [shift={(250,709)}, rotate = 360] [color={rgb, 255:red, 0; green, 0; blue, 0 }  ][line width=0.75]    (10.93,-3.29) .. controls (6.95,-1.4) and (3.31,-0.3) .. (0,0) .. controls (3.31,0.3) and (6.95,1.4) .. (10.93,3.29)   ;
\draw    (76,409) -- (76,480) -- (123,480) ;
\draw [shift={(125,480)}, rotate = 180] [color={rgb, 255:red, 0; green, 0; blue, 0 }  ][line width=0.75]    (10.93,-3.29) .. controls (6.95,-1.4) and (3.31,-0.3) .. (0,0) .. controls (3.31,0.3) and (6.95,1.4) .. (10.93,3.29)   ;
\draw    (56,410) -- (55,582) -- (122,582) ;
\draw [shift={(124,582)}, rotate = 180] [color={rgb, 255:red, 0; green, 0; blue, 0 }  ][line width=0.75]    (10.93,-3.29) .. controls (6.95,-1.4) and (3.31,-0.3) .. (0,0) .. controls (3.31,0.3) and (6.95,1.4) .. (10.93,3.29)   ;
\draw    (37,408) -- (37,699) -- (122,698.02) ;
\draw [shift={(124,698)}, rotate = 179.34] [color={rgb, 255:red, 0; green, 0; blue, 0 }  ][line width=0.75]    (10.93,-3.29) .. controls (6.95,-1.4) and (3.31,-0.3) .. (0,0) .. controls (3.31,0.3) and (6.95,1.4) .. (10.93,3.29)   ;

\draw (60,130) node [anchor=north west][inner sep=0.75pt]   [align=center] {{\huge $\textbf{{\fontfamily{lmr}\selectfont x}}\displaystyle \sim {\fontfamily{lmr}\selectfont \textit{X}}$}};
\draw (440,104) node [anchor=north west][inner sep=0.75pt]   [align=left] {{\LARGE $\textbf{{\fontfamily{lmr}\selectfont w}}\displaystyle \sim {\fontfamily{lmr}\selectfont \textit{W}}$}};
\draw (394,169) node [anchor=north west][inner sep=0.75pt]   [align=left] {{\Large $\displaystyle \mathcal{A} \displaystyle \Rightarrow \textbf{{\fontfamily{lmr}\selectfont wx}} \displaystyle ^{\intercal } \displaystyle \Rightarrow \textit{{\fontfamily{lmr}\selectfont f}}$}};
\draw (244,70) node [anchor=north west][inner sep=0.75pt]   [align=left] {{\Large ${\fontfamily{lmr}\selectfont \textbf{q}}\displaystyle \sim {\fontfamily{lmr}\selectfont \textit{Q}}$}};
\draw (278,121) node [anchor=north west][inner sep=0.75pt]  [font=\Large] [align=left] {$\displaystyle \mathcal{A}$};
\draw (244,161) node [anchor=north west][inner sep=0.75pt]  [font=\small] [align=left] {{\Large {\fontfamily{lmr}\selectfont \textit{f}(\textbf{wx}}$\displaystyle ^{\intercal }{\fontfamily{lmr}\selectfont )}$ $ $}};
\draw (78,39) node [anchor=north west][inner sep=0.75pt]   [align=left] {{\fontfamily{ptm}\selectfont {\Large User}}};
\draw (455,36) node [anchor=north west][inner sep=0.75pt]   [align=left] {{\fontfamily{ptm}\selectfont {\Large Server}}};
\draw (79,364) node [anchor=north west][inner sep=0.75pt]   [align=left] {{\Large ${\fontfamily{lmr}\selectfont \textbf{q}}\displaystyle \sim {\fontfamily{lmr}\selectfont \textit{Q}}$}};
\draw (47,313) node [anchor=north west][inner sep=0.75pt]   [align=left] {{\fontfamily{ptm}\selectfont {\large Public forum}}};
\draw (459,403) node [anchor=north west][inner sep=0.75pt]   [align=left] {{\fontfamily{ptm}\selectfont {\Large Server}}};
\draw (164,415) node [anchor=north west][inner sep=0.75pt]   [align=left] {{\fontfamily{ptm}\selectfont {\large User}}};
\draw (164,636) node [anchor=north west][inner sep=0.75pt]   [align=left] {{\fontfamily{ptm}\selectfont {\large User}}};
\draw (164,523) node [anchor=north west][inner sep=0.75pt]   [align=left] {{\fontfamily{ptm}\selectfont {\large User}}};
\draw (294,432) node [anchor=north west][inner sep=0.75pt]  [font=\large] [align=left] {$\displaystyle \mathcal{A}_{1}$};
\draw (274,490) node [anchor=north west][inner sep=0.75pt]  [font=\footnotesize] [align=left] {{\large {\fontfamily{lmr}\selectfont \textit{f}(\textbf{wx}}$\displaystyle _{1}^{\intercal }${\fontfamily{lmr}\selectfont )}$ $}};
\draw (274,598) node [anchor=north west][inner sep=0.75pt]  [font=\footnotesize] [align=left] {{\large {\fontfamily{lmr}\selectfont \textit{f}(\textbf{wx}}$\displaystyle _{2}^{\intercal }${\fontfamily{lmr}\selectfont )}$ $}};
\draw (294,540) node [anchor=north west][inner sep=0.75pt]  [font=\large] [align=left] {$\displaystyle \mathcal{A}_{2}$};
\draw (294,652) node [anchor=north west][inner sep=0.75pt]  [font=\large] [align=left] {$\displaystyle \mathcal{A}_{3}$};
\draw (274,710) node [anchor=north west][inner sep=0.75pt]  [font=\footnotesize] [align=left] {{\large {\fontfamily{lmr}\selectfont \textit{f}(\textbf{wx}}$\displaystyle _{3}^{\intercal }${\fontfamily{lmr}\selectfont )}$ $}};
\draw (137,459) node [anchor=north west][inner sep=0.75pt]   [align=left] {{\Large $\textbf{{\fontfamily{lmr}\selectfont x}}\displaystyle _{1}$}{ \Large $\displaystyle \sim {\fontfamily{lmr}\selectfont \textit{X}}$}};
\draw (137,570) node [anchor=north west][inner sep=0.75pt]   [align=left] {{\Large $\textbf{{\fontfamily{lmr}\selectfont x}}\displaystyle _{2}$}{ \Large $\displaystyle \sim {\fontfamily{lmr}\selectfont \textit{X}}$}};
\draw (137,680) node [anchor=north west][inner sep=0.75pt]   [align=left] {{\Large $\textbf{{\fontfamily{lmr}\selectfont x}}\displaystyle _{3}$}{ \Large $\displaystyle \sim {\fontfamily{lmr}\selectfont \textit{X}}$}};
\draw (442,514) node [anchor=north west][inner sep=0.75pt]   [align=left] {{\LARGE $\textbf{{\fontfamily{lmr}\selectfont w}}\displaystyle \sim {\fontfamily{lmr}\selectfont \textit{W}}$}};
\draw (387,612) node [anchor=north west][inner sep=0.75pt]   [align=left] {{\Large $\displaystyle \mathcal{A}_i \displaystyle \Rightarrow \textbf{{\fontfamily{lmr}\selectfont wx}} \displaystyle _i ^{\intercal } \displaystyle \Rightarrow \textit{{\fontfamily{lmr}\selectfont f}}$}};

\end{tikzpicture}
 \ \\ \ \\ 
\caption{An illustration of the information-theoretic PI problem. (top) In a two-party setting, a server holds a weight vector $\boldw \sim W$ and a user holds data $\boldx \sim X$. The server sends a query $\boldq$ to the user, and the user replies with an answer $\cA$. This answer is used by the server to extract the signal $\boldw\boldx^\intercal$, which is then fed into a machine learning model~$f$. The inference~$f(\boldw\boldx^\intercal)$ is sent back to the user. (bottom) Our scheme can also be used in a multi-user single-server case, where each user has their own data $\boldx_i \sim X$. The query $\boldq$ is published in a public forum, from which it is downloaded by the users. The protocol then proceeds as in the single-user case. The privacy guarantees from the two-party setting remain.}
    \label{fig:system}
\end{figure}
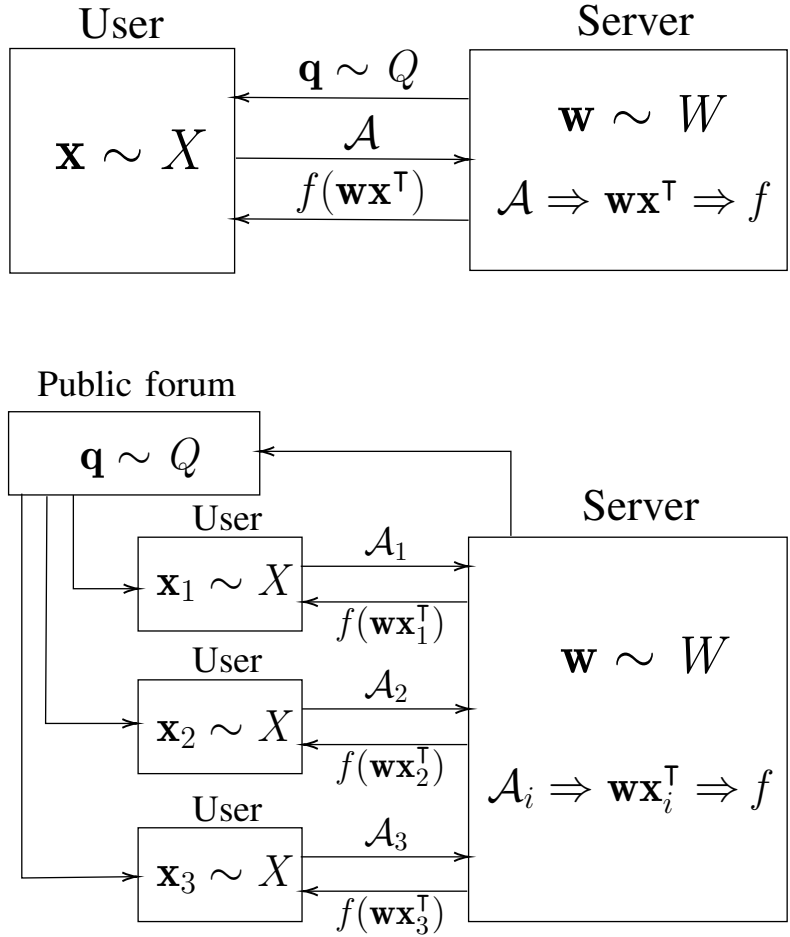


The merit of query-answer protocols for PI problems is measured by the following quantities:

\begin{itemize}
    \item \textbf{Publication cost}, i.e., the number of bits $d$ in the query $\boldq$. Notice that this number might change as a function of the value of~$Q$; finding the optimal expected value of~$d$ is a \textit{source coding} problem, and this value is lower bounded by~$H(Q)$ according to the well-known source coding theorem by Shannon~\cite[Th.~5.3.1]{CovThom06}.
    \item \textbf{Server-privacy}, measured by the mutual information $I(W; Q)$ in bits.
    \item \textbf{User-privacy}, measured by the dimension of the subspace on which $\boldx$ is revealed, i.e., the parameter\footnote{More precisely, the parameter $\dim(\RSpan(\{\boldv_i\}_{i=1}^{\ell}))$, but for simplicity of presentation we consider the worst-case assumption, i.e., that the vectors $\boldv_i$ are independent.} $\ell$. 
\end{itemize}

Ideally, we would like to minimize all three quantities simultaneously. However, it will be shown in the sequel that there is an inherent tradeoff between server-privacy and user-privacy, assuming that the decoding part of the protocol is done by using a polynomial whose coefficients depend only on $Q$ and not on $X$. Therefore, our goal is to attain the tradeoff while minimizing publication cost as much as possible. 

\subsection{Our contributions}
We present two protocols for the PI problem where the weight vector $\boldw$ is taken from~$\{\pm 1\}^n$, and both of these protocols are optimal with respect to the privacy tradeoff as well as publication cost\footnote{In the appendix, we also come up with a protocol for weight vectors taking values from complex roots of unity of any order, and further modify this scheme to make it applicable to $\boldw \in \{0, \pm 1\}^n$. However, the optimality of the resulting $\{0, \pm 1\}$ protocol is unclear.}. 
We provide two equivalent protocols since their generalizations, which are used in later parts of the paper, do not coincide. For the setting where multiple inner products with weight vectors from $\{\pm 1\}^n$ are to be jointly retrieved, we provide a protocol which achieves an optimal privacy tradeoff in most cases, albeit with a sub-optimal publication cost.

The paper then investigates PI problems where weight vectors are restricted to special types of finite real sets, namely \textit{perfect sets} and \textit{hard sets}. These set definitions employ the recently defined additive combinatorial notion of \textit{coefficient complexity}, which quantifies the algebraic structure of finite sets of real numbers as linear combinations of~$\{\pm1\}$ vectors, and was introduced in~\cite{ramkumar2023coefficient} for applications in distributed computation.
Roughly speaking, a coefficient set~$A$ is called \textit{perfect} if it admits the most compact representation as a linear combination of~$\{\pm1\}$ vectors, and it is called \textit{hard} if any such representation always contains a maximum number of~$\{\pm1\}$ vectors. Using the joint retrieval protocol, we provide a protocol attaining the privacy tradeoff for the perfect set problem in most cases. For hard sets, we present a sub-optimal protocol based on Hadamard matrices. All optimality results mentioned above originate from the privacy tradeoff bound and a bound on the expected publication cost that we derive for arbitrary finite sets.

\subsection{Notations}
All vectors in this work are row vectors, unless stated otherwise. Throughout this paper, logarithms are in base~$2$, the notation $[n] \triangleq \{1, \ldots, n\}$ is used, and we refer to~$\bF_2$ as the binary field represented by $\{\pm 1\}$, i.e., $-1$ represents the Boolean ``one,'' and $1$ represents the Boolean ``zero.'' 
We use $\oplus, \odot$ for addition and multiplication operations, respectively, in $\bF_2$ and $+, \cdot$ for their counterparts in $\bR$. Notice that since the $\{\pm 1\}$-representation of $\bF_2$ is used, the $\oplus$ operation between elements of $\bF_2$ is identical to the $\cdot$ operation between elements of $\bR$, i.e., $x \oplus y = x \cdot y$ for all $x, y \in \{\pm 1\}$. For all~$x \in \{\pm 1\}$, we have that $1 \odot x = 1$ and $-1 \odot x = x$. We use~$\bigoplus$ for summation of multiple entries in~$\bF_2$ and~$\Sigma$ for summation of multiple entries in~$\bR$. 

For a real vector $\boldu = (u_1, \ldots, u_n)$ and a real number $\gamma$, we have $\gamma \cdot \boldu = (\gamma \cdot u_1, \ldots, \gamma \cdot u_n)$. Similarly, for $\gamma \in \bF_2$ and $\boldu \in \bF_{2}^n$, $\gamma \oplus \boldu = (\gamma \oplus u_1, \ldots, \gamma \oplus u_n)$ and $\gamma \odot \boldu = (\gamma \odot u_1, \ldots, \gamma \odot u_n)$. By $(1, \boldu)$, we denote the $(n+1)$-length vector whose first entry is $1$ and $(i+1)$-th entry is $u_i$ for all $i \in [n]$. A collection of sets $\{A_1, \ldots, A_t\}$ is called a partition of $[n]$ if $A_i \cap A_j = \emptyset$ for every $i \ne j$, $A_i \ne \emptyset$ for all $i \in [t]$, and $\cup_{i=1}^{t} A_i = [n]$. For a partition~$\{A_1,\ldots,A_t\}$ of~$[n]$, we denote by $A_{i,j}$ the $j$-th element of $A_i$ and by $u_{A_{i,j}}$ the entry of $\boldu$ indexed by $A_{i,j}$. For a set $S \subseteq [n]$, we let~$\1_{S}\in\bF_2^n$ be the characteristic vector of~$S$, i.e., $(\1_{S})_j=-1$ if~$j\in S$, and~$(\1_{S})_j=1$ otherwise. For a vector~$\boldu$ of length~$n$ and a set~$S\subseteq[n]$, we let~$\boldu\vert_{S}$ be the vector of length~$|S|$ which contains the entries of~$\boldu$ indexed by~$S$. 
For a set of vectors $\cU$, we use $\cU\vert_S$ to denote the set $\{\boldu\vert_S~\vert~\boldu \in \cU\}$. We naturally extend the notation~$\vert_{S}$ to matrices; that is, for a matrix~$\boldU$, we denote $\boldU\vert_{S}$ as the submatrix of~$\boldU$ formed by the columns whose indices are in $S$.

\subsection{Justification for the user-privacy parameter}\label{section:justificationforell}
As mentioned earlier, for simplicity and generality of our treatment of the PI problem, we choose to quantify the loss of privacy from the user's side (i.e., of the data random variable~$X$) as the number of linear projections sent to the server throughout the protocol. This metric is well justified by (the converse of) known information-theoretic results in the area of linear dimension reduction. In linear dimension reduction, one wishes to project data linearly so as to preserve as much of its informativeness as possible. One popular way of quantifying this informativeness is via mutual information, i.e., via choosing a projection matrix which maximizes the mutual information between the original data and its projected version. This is commonly referred to as the Infomax principle~\cite{linsker1988self}. 
\newline\indent It is well known~(\cite[Sec.~3.2.1,~3.2.2]{deco1996information}, originally~\cite{plumbley1988information,plumbley1991information}) that for Gaussian distributions, the projection which attains the Infomax principle, i.e., maximizes the mutual information between the data and its projection, is given by the~$\ell$ eigenvectors of the covariance matrix which correspond to the~$\ell$ largest eigenvalues, and the~$\ell$ eigenvectors are known as principal components. For all other distributions, the~$\ell$ eigenvectors of the best Gaussian approximation (i.e., a Gaussian distribution with identical first and second moments) are at least as informative as the Infomax solution.
\newline\indent
In the context of PI, however, the goal is reversed to that of dimension reduction, in the sense that one wishes the projection to provide \textit{as little} information about the original data as possible. Therefore, taking~$\ell$ as a privacy metric is a worst-case and monotone quantification of privacy loss; larger~$\ell$ means greater privacy loss, and vice-versa (the exact mutual information is given by the value of the Infomax solution, and is often hard to compute).
\newline\indent
In practical terms, given the projections of~$\boldx$ onto~$\boldv_1,\ldots,\boldv_\ell$, the server may deduce that~$\boldx$ lies in some subspace, i.e., in the solution space of a (potentially non-homogeneous) linear system. Data distributions which lie in a small subspace can be a priori projected onto that subspace prior to the execution of our protocol, and privacy loss is then better analyzed with respect to the dimension of the projected data rather than that of the original data.

\section{The~$\{\pm 1\}$ PI problem}\label{section:pmone}

In this section, we consider the case of a weight vector $\boldw$ with $\pm 1$ entries. A server holds a weight vector $\boldw \in \{\pm 1\}^n$, randomly chosen from $W = \textrm{Unif}(\{\pm 1\}^n)$, and a user holds a real-valued data vector $\boldx$ of the same length. The server aims to retrieve the signal $\boldw\boldx^\intercal$ (computed over $\bR$) while safeguarding the privacy of the server and the user. 



It will be shown in the sequel that $I(W; Q) + \ell$ is bounded from below. Therefore, we seek to achieve this lower bound while minimizing publication cost $d$ as much as possible, by devising protocols of the aforementioned query-answer form.

\subsection{Our protocols}\label{section:pmonealgos}


We now present two query-answer protocols for the $\{\pm 1\}$ PI problem called \textit{the coset protocol} and \textit{the random key protocol}. It will be shown later that they both achieve the privacy lower bound and have minimal publication cost. Before providing the protocols, we introduce some requisite notations. As noted above, we employ the $\{\pm1\}$-representation of $\bF_2$. For a privacy parameter $t \in [n]$ and a partition $\cS = \{S_1, \ldots, S_t\}$ of $[n]$, 
let $V_\cS$ be the~$\bF_2$-span of the characteristic vectors of the sets~$S_i$, i.e., $V_\cS \triangleq \FtwoSpan\{\1_{S_i}\}_{i=1}^t$. In addition, let $M \in \bF_{2}^{(n-t)\times n}$ be a parity-check matrix for $V_\cS$, i.e.,~$V_\cS$ is the right kernel of~$M$ over~$\bF_2$, and let $B(M, \boldb)$ be some deterministic algorithm that finds a solution $\boldy$ to the equation $M \odot \boldy^\intercal = \boldb$ over $\bF_{2}$. Both the partition~$\cS$ and the algorithm~$B$ are known to both parties.  

\subsubsection*{The coset protocol}
In this protocol, the server communicates to the user the identity of the coset of~$V_\cS$ containing~$\boldw$ by computing a syndrome, which in turn identifies a shift vector~$\boldu$ that defines that coset (i.e., $\boldw\in V_\cS \oplus \boldu$). Then, the user defines vectors~$\{\boldv_i\}_{i=1}^{t}$ based on $\boldu$ and uses them to respond to the query. Formally, the coset protocol is as follows:

\begin{enumerate}[label=\roman*)]
    \item The server:
    \begin{enumerate}[label=\alph*)]
    \item Publishes $\boldq \triangleq M \odot \boldw^\intercal$.
    \item Defines $\boldu \triangleq B(M, \boldq)$ and keeps it private.
    \item Finds the unique $\ell_1, \ldots, \ell_t \in \bF_{2}$ such that $$\boldu = \boldw \oplus \bigoplus_{i=1}^t(\ell_i \odot \1_{S_i})$$ and keeps them private. The $\ell_i$'s exist because $M \odot \boldu^\intercal = M \odot \boldw^\intercal$, and thus $\boldu \in V_\cS \oplus \boldw = \operatorname{Span}_{\bF_{2}}\{\1_{S_i}\}_{i=1}^t \oplus \boldw$.
    \end{enumerate}
    \item The user:
    \begin{enumerate}[label=\alph*)]
    \item Defines $\boldu \triangleq B(M, \boldq)$. This is the same $\boldu$ as the one defined by the server since $B$ is deterministic.
    \item Sends~$\cA = \{(\boldu\vert_{S_i})(\boldx\vert_{S_i})^\intercal\}_{i=1}^t$ to the server (computed over~$\bR$).
    \end{enumerate}
    \item The server combines the values in $\cA$ to compute and output $$\sum_{i=1}^{t}\ell_i\cdot(\boldu\vert_{S_i})(\boldx\vert_{S_i})^\intercal = \boldw\boldx^\intercal.$$   
\end{enumerate}
The last computation is correct because $$\boldu\vert_{S_i} = \boldw\vert_{S_i} \oplus \ell_i = \ell_i \cdot (\boldw\vert_{S_i}),$$
and, recalling $\ell_i \in \{\pm 1\}$ for all $i$, we have
\begin{align*}
\sum_{i=1}^t\ell_i\cdot(\boldu\vert_{S_i})(\boldx\vert_{S_i})^\intercal &= \sum_{i=1}^t\ell_i\cdot(\ell_i\cdot (\boldw\vert_{S_i}))(\boldx\vert_{S_i})^\intercal \\ &=\sum_{i=1}^t(\boldw\vert_{S_i})(\boldx\vert_{S_i})^\intercal \\ &=\boldw\boldx^\intercal.
\end{align*}

\subsubsection*{The random key protocol}
In this protocol, we partition $\boldw$ according to $\cS$ and mask the identity of each resulting segment of $\boldw$ with a random $\{\pm 1\}$ key. The server shares the masked (encoded) version of $\boldw$ as a query, the user returns an answer accordingly, and the server combines elements of the answer to retrieve the desired inner product using the random keys. The protocol works as follows.

\begin{enumerate}[label=\roman*)]
    \item The server publishes $\ell_i \cdot \boldw\vert_{S_i}$ for every $i \in [t]$, where each~$\ell_i \in \{\pm 1\}$ is a random key and is kept private. 
    \item The user sends $\{(\ell_i \cdot \boldw\vert_{S_i})(\boldx\vert_{S_i})^\intercal\}_{i=1}^{t}$ to the server (computed over~$\bR$).
    \item The server computes and outputs $$\sum_{i=1}^{t} \ell_i \cdot(\ell_i\cdot \boldw\vert_{S_i})(\boldx\vert_{S_i})^\intercal = \boldw\boldx^\intercal.$$   
\end{enumerate}

This protocol can be further improved if one chooses the keys in a judicious manner. Specifically, we choose $\ell_i \triangleq w_{S_{i,1}}$. Since the first bit of $w_{S_{i,1}}\cdot \boldw\vert_{S_i}$ is always $1$, the server may omit it from publication, and the user may simply prefix a ``$1$'' bit locally before proceeding with computation. The improved scheme works as follows.

\begin{enumerate}[label=\roman*)]
    \item The server publishes $w_{S_{i,1}}\cdot \boldw\vert_{S_i \setminus \{S_{i,1}\}}$ for every $i \in [t]$.
    \item The user sends $\{(\hat{\boldw}\vert_{S_i})(\boldx\vert_{S_i})^\intercal\}_{i=1}^{t}$ to the server (computed over~$\bR$), where $\hat{\boldw} \in \{\pm 1\}^n$ is defined as $$\hat{\boldw}\vert_{S_i} \triangleq (1, w_{S_{i,1}}\cdot \boldw\vert_{S_i \setminus \{S_{i,1}\}}).$$
    \item The server computes and outputs \begin{align*} \sum_{i=1}^{t}w_{S_{i,1}}\cdot(\hat{\boldw}\vert_{S_i})(\boldx\vert_{S_i})^\intercal &= \sum_{i=1}^{t}w_{S_{i,1}}\cdot(1, w_{S_{i,1}}\cdot \boldw\vert_{S_i \setminus \{S_{i,1}\}})(\boldx\vert_{S_i})^\intercal \\ &= \sum_{i=1}^{t}(w_{S_{i,1}}, \boldw\vert_{S_i \setminus \{S_{i,1}\}})(\boldx\vert_{S_i})^\intercal \\ &=\sum_{i=1}^t(\boldw\vert_{S_i})(\boldx\vert_{S_i})^\intercal =\boldw\boldx^\intercal.   
    \end{align*}
\end{enumerate}

\subsection{Merit of the protocols}

Next we analyze the coset protocol and the improved random key protocol in terms of publication cost, server-privacy and user-privacy. It will be later shown in Section~\ref{section:pmoneoptimality} that both these protocols are optimal in terms of the privacy tradeoff and publication cost.

\begin{theorem}
    For any privacy parameter $t \in [n]$, both the coset protocol and the improved random key protocol for the $\{\pm 1\}$ PI problem have publication cost $d = n - t$, server-privacy $I(W; Q) = n - t$, and user-privacy $\ell = t$. 
\end{theorem}

\begin{proof}
    For the coset protocol, publication cost and user-privacy claims are immediate by definition. Because $$I(W; Q) = H(W) - H(W|Q) = n - H(W|Q),$$ it suffices to show that $H(W|Q) = t$. Notice that disclosing $Q = \boldq$ reveals the identity of the coset of $V_\cS$ containing~$\boldw$. Since $\boldw$ is uniformly distributed over $\{\pm 1\}^n$, and since cosets of $V_\cS$ are of size $2^t$, it follows that $W|Q = \boldq$ is uniform over a set of size $2^t$, and hence $H(W|Q = \boldq) = t$. Thus, 
    \begin{align*}
        H(W|Q) &= \sum_{\boldq}\Pr(\boldq)H(W|Q=\boldq) \\ &= H(W|Q = \boldq) = t.
    \end{align*}

    For the improved random key protocol, the user-privacy claim is immediate by definition. With the informed choice of keys, the server saves one bit of publication cost in every $\boldw\vert_{S_i}$, therefore lowering $d$ to $n - t$. To compute $I(W; Q)$, notice that after the query is published, from the user's perspective, the only source of uncertainty of $\boldw$ comes from the keys, and there are a total of $2^t$ possibilities for their values, all of them equiprobable. It follows similarly that $W|Q = \boldq$ is uniform over a set of size $2^t$, and 
    \begin{align*}
        I(W; Q) = H(W) - H(W|Q) &= n - t. \qedhere
    \end{align*} 
\end{proof}
In Appendix \ref{section:rootsofunity}, we generalize the improved random key protocol to the PI problem where the weights are restricted to roots of unity of any order. This roots-of-unity protocol is then modified to obtain a protocol for the $\{0, \pm 1\}$ PI problem.

\section{The PI problem for joint retrieval of multiple signals}\label{section:jointretrieval}
This section addresses the problem of retrieving multiple signals jointly. The server now possesses multiple $\{\pm 1\}$-vectors and seeks to retrieve their respective inner products with the data vector simultaneously. Formally, the server holds a set of $m \ge 2$ weight vectors $\boldw^{(1)}, \ldots, \boldw^{(m)} \in \{\pm 1\}^{n}$, chosen independently at random from $\textrm{Unif}(\{\pm 1\}^n)$, and aims to retrieve a set of inner products $\{\boldw^{(i)}\boldx^\intercal\}_{i=1}^{m}$.

Consider the field $\bF_{2^m}$ and identify the elements of this field with $\{\pm 1\}$-vectors of length~$m$. In this section, we use $\oplus, \odot$ for addition and multiplication operations, respectively, in $\bF_{2^m}$. Under the $\{\pm 1\}$-representation of $\bF_{2^m}$, for all $x, y \in \{\pm 1\}^m$, notice that $x \oplus y$ is identical to the element-wise product of $x$ and $y$ over $\bR$. Let $\boldW \in \{\pm 1\}^{m \times n}$ be the matrix whose rows are the~$m$ vectors $\boldw^{(1)}, \ldots, \boldw^{(m)}$. It follows that obtaining the required set of signals is equivalent to computing $\boldW\boldx^\intercal$. Since each column of the matrix $\boldW$ corresponds to an element in $\bF_{2^m}$, it follows that $\boldW$ can be represented as a vector in $\bF_{2^m}^n$, which we denote by $\boldw$. Note that $\boldw$ can be viewed as being randomly chosen from a uniform distribution $W = \textrm{Unif}(\bF_{2^m}^n)$. For $\boldu = (u_1,\ldots,u_n) \in \bF_{2^m}^n$, the~$\{\pm 1\}$-matrix representation of~$\boldu$ is the matrix~$U \in \{\pm 1\}^{m \times n}$ having the $m$-length $\{\pm 1\}$-vector corresponding to $u_i$ as the $i$-th column, for all $i \in [n]$.

We now give an overview of our protocol for the joint retrieval PI problem. The protocol is motivated by the coset protocol for the~$\{\pm 1\}$ PI problem presented in Section \ref{section:pmonealgos}. We first establish that $\bF_{2^m}$ can be partitioned into pairwise disjoint subsets such that each subset is contained in a real subspace of low dimension. A key concept introduced shortly is that of a \textit{good} vector for a certain partition of $[n]$. A vector $\boldw \in \bF_{2^m}^n$ is considered \textit{good} for a partition $\cS = \{S_1, \ldots, S_t\}$ of $[n]$ if, for every segment $\boldw\vert_{S_i}$, there is a subset of $\bF_{2^m}$ (in the aforementioned partition of the field) to which all the entries of $\boldw\vert_{S_i}$ belong; this property ensures that the entries lie in a subspace of low dimension. Under appropriate choices of parameters, for every $\boldw \in \bF_{2^m}^n$ there is a partition of $[n]$ for which $\boldw$ is good. The server needs to find such a partition of $[n]$ and communicate its identity. 

We define $V_\cS$ as the span of the characteristic vectors of the sets in the partition~$\cS$, similar to what is done in Section \ref{section:pmonealgos}, but over~$\bF_{2^m}$ instead of over~$\bF_2$. The server publishes the identity of the coset of $V_\cS$ containing $\boldw$ by computing a syndrome. The server and the user then agree on a shift vector~$\boldu$ that belongs to the same coset as $\boldw$, using a deterministic algorithm that takes the syndrome as input and returns the shift vector as output. 
The user finds a basis for every~$\boldu\vert_{S_i}$ and sends responses accordingly. The server uses these responses to compute $\boldW\boldx^\intercal$. The correctness and merit of the protocol hinge on Lemma~\ref{lemma:mminuslogq} which follows, stating that every segment $\boldu\vert_{S_i}$ has the same row span as the corresponding segment~$\boldw\vert_{S_i}$.

\subsection{Preliminaries}
  Let $q$ be an integer power of two such that $1 \le q \le 2^{m-1}$. As previously mentioned, we represent elements of $\bF_{2^m}$ using $\{\pm 1\}$-vectors of length~$m$.
  We partition $\bF_{2^m}$ into $q$ pairwise disjoint sets $F_1, \ldots, F_q$ of equal size, i.e., $\bF_{2^m} = \cup_{i=1}^{q} F_i$ with $|F_i| = \frac{2^m}{q}$ for every $i \in [q]$, such that the~$\bR$-span of each set is an~$\bR$-subspace of dimension $\log (\frac{2^m}{q}) = m - \log q$. Such a partition exists due to the following lemma.


\begin{lemma}\label{lemma:fieldpartition}
    Let~$q$ ($1 \le q \le 2^{m-1}$) be an integer power of two. Then, there exists a partition $\mathcal{F} = \{F_1, \ldots, F_q\}$ of the extension field~$\bF_{2^m}$ into~$q$ pairwise disjoint subsets of size $\frac{2^m}{q}$ such that $$\dim(\operatorname{Span}_{\bR} (F_i)) = m - \log q$$ for every $i \in [q]$.
\end{lemma}

\begin{proof}
    Let $p \triangleq m - \log q$, let $\{A_1, \ldots, A_p\}$ be a partition of~$[m]$ into~$p$ distinct subsets, and let $$F_1 \triangleq \operatorname{Span}_{\bF_2}\{\1_{A_i}\}_{i=1}^p.$$ It follows that $|F_1| = 2^p = \frac{2^m}{q}$, and since cosets of $F_1$ are also of size $2^p$, there are a total of $\frac{2^m}{2^p} = q$ cosets, including~$F_1$ itself. We fix~$F_2,\ldots,F_q$ as the remaining cosets of~$F_1$, and 
    it remains to show that $\dim(\operatorname{Span}_{\bR} (F_i)) = m - \log q$ for all~$i \in [q]$. 

    For a given~$i \in [q]$, let $\boldu \in \bF_2^m$ be a shift vector that defines $F_i$, i.e., $$F_i = F_1 \oplus \boldu.$$ For every vector $\boldz \in F_i$, there exists $\boldy \in F_1$ for which $\boldz = \boldy \oplus \boldu$. It follows from the definition of $F_1$ that for all $j \in [p]$, there is some $\alpha_j \in \bF_2$ such that $y_g = \alpha_j$ for all $g \in A_j$.
    As mentioned before, since the $\{\pm 1\}$-representation of $\bF_2$ is used, the addition operation between elements of $\bF_2$ is identical to the multiplication operation between elements of $\bR$. This results in $$\boldz\vert_{A_j} = \boldy\vert_{A_j} \oplus \boldu\vert_{A_j} = \alpha_j\boldu\vert_{A_j},$$ with $\alpha_j\in\{\pm1\}$ over $\bR$, for all $j \in [p]$.
    Therefore, the vectors in $F_i$ are contained in an $\bR$-subspace spanned by the $p$ vectors $\boldv^{(1)}, \ldots, \boldv^{(p)} \in \bR^m$ defined below:
    \begin{align*} 
        v^{(j)}_g = \begin{cases}
            u_g & g \in A_j, \\ 0 & g \in [m] \setminus A_j.
        \end{cases}
    \end{align*}
    Hence, $\dim(\operatorname{Span}_{\bR} (F_i)) = p = m - \log q$ for all~$i \in [q]$.
\end{proof}

As stated in the overview, a key component in the upcoming algorithm is finding a  partition~$\cS=\{S_1,\ldots,S_t\}$ of $[n]$ such that for every~$i \in [t]$, the entries of~$\boldw\vert_{S_i}$ belong to a real subspace of~$\bR^m$ of low dimension. Lemma \ref{lemma:fieldpartition} establishes that a partition of $[n]$ such that all elements of~$\boldw\vert_{S_i}$ come from the same $F_j$ would suffice. In such a partition~$\cS$, since $$\sum_{i=1}^t(\boldW\vert_{S_i})(\boldx\vert_{S_i})^\intercal=\boldW\boldx^\intercal,$$ and since for every~$i \in [t]$, $$\operatorname{rank}_\bR(\boldW\vert_{S_i})\le m-\log q,$$ it follows that projecting each~$\boldx\vert_{S_i}$ on a maximal set of linearly independent rows of~$\boldW\vert_{S_i}$ suffices to compute all~$m$ entries of~$\boldW\boldx^\intercal$. To this end, the server should communicate to the user a basis for the row span of each~$\boldW\vert_{S_i}$ in a private way. Next, we formally define what a \textit{good} vector is for a partition of~$[n]$.

\begin{definition}
    We call a vector~$\boldw \in \bF_{2^m}^n$ \textit{good} for a partition~$\mathcal{S}$ of~$[n]$ if for any $i \in [t]$ there exists some $j \in [q]$ such that all entries of $\boldw\vert_{S_i}$ come from the subset $F_j$ of $\bF_{2^m}$ in the partition of the field according to Lemma~\ref{lemma:fieldpartition}. That is, there exists a mapping~$g:[t]\to[q]$ such that
\begin{align*}
    \{ w_s\vert s\in S_i \}\subseteq F_{g(i)}
\end{align*}
for every~$i\in[t]$. 
\end{definition}
The following lemma establishes for every~$\boldw$ the existence of a partition of $[n]$ for which $\boldw$ is good.

\begin{lemma}
    Let $q$ ($1 \le q \le 2^{m-1}$) be an integer power of two. For each~$\boldw \in \bF_{2^m}^n$, there exists a partition of $[n]$ into $t$ subsets for which $\boldw$ is good as long as $q \le t$.
\end{lemma}
\begin{proof}
    Suppose $q \le t$. Consider the partition  $\hat{\cS} = \{\hat{S}_1, \ldots, \hat{S}_q\}$ of $[n]$ defined as 
    \begin{align*}
        \hat{S}_i = \{j\vert w_j \in F_i\}.
    \end{align*}
    It can be seen that $\boldw$ is good for any partition $\mathcal{S} = \{S_1, \ldots, S_t\}$ of $[n]$ such that for every $i \in [t]$, there exists some $j \in [q]$ for which $S_i \subseteq \hat{S}_j$.
\end{proof}

\begin{example}
    Let $n = 9$, $m = q = 4$, $t = 5$. We have $m - \log q = 2$. First, we demonstrate that there exists an appropriate partition of the field $\bF_{2^4}$. Let $\{A_1 = \{1, 2\}, A_2 = \{3, 4\}\}$ be a partition of $[m] = [4]$. Then,
    \begin{align*}
        F_1 &= \operatorname{Span}_{\bF_2}\{[-1, -1, 1, 1], [1, 1, -1, -1]\} \\ &= \{[1, 1, 1, 1], [1, 1, -1, -1], [-1, -1, 1, 1], [-1, -1, -1, -1]\},
    \end{align*}
    and therefore 
    \begin{align*}
        F_2 &= \{[-1, 1, 1, 1], [-1, 1, -1, -1], [1, -1, 1, 1], [1, -1, -1, -1]\}, \\
        F_3 &= \{[1, 1, -1, 1], [1, 1, 1, -1], [-1, -1, -1, 1], [-1, -1, 1, -1]\}, \\
        F_4 &= \{[1, -1, -1, 1], [1, -1, 1, -1], [-1, 1, -1, 1], [-1, 1, 1, -1]\}
    \end{align*}
    are cosets of $F_1$. It can be easily checked that $F_1$, $F_2$, $F_3$ and $F_4$ are each contained in an $\bR$-subspace of dimension $2$, and that $\bF_{2^4} = \cup_{i=1}^{4} F_i$. Now, let 
    \begin{align*}
        \boldw^{(1)} &= [-1, \phantom{-}1, -1, \phantom{-}1, \phantom{-}1, -1, -1, -1, \phantom{-}1], \\
        \boldw^{(2)} &= [-1, \phantom{-}1, \phantom{-}1, \phantom{-}1, -1, \phantom{-}1, \phantom{-}1, \phantom{-}1, \phantom{-}1], \\
        \boldw^{(3)} &= [\phantom{-}1, \phantom{-}1, \phantom{-}1, \phantom{-}1, \phantom{-}1, -1, \phantom{-}1, -1, -1], \\
        \boldw^{(4)} &= [\phantom{-}1, \phantom{-}1, -1, -1, -1, \phantom{-}1, \phantom{-}1, -1, -1],
    \end{align*}
    which means
    \begin{align*}
        \boldW =
        \begin{bmatrix}
            -1 & \phantom{-}1 & -1 & \phantom{-}1 & \phantom{-}1 & -1 & -1 & -1 & \phantom{-}1 \\
            -1 & \phantom{-}1 & \phantom{-}1 & \phantom{-}1 & -1 & \phantom{-}1 & \phantom{-}1 & \phantom{-}1 & \phantom{-}1 \\
            \phantom{-}1 & \phantom{-}1 & \phantom{-}1 & \phantom{-}1 & \phantom{-}1 & -1 & \phantom{-}1 & -1 & -1 \\ \phantom{-}1 & \phantom{-}1 & -1 & -1 & -1 & \phantom{-}1 & \phantom{-}1 & -1 & -1
        \end{bmatrix}.
    \end{align*}
    It is straightforward to find a partition of~$[n] = [9]$ into~$t = 5$ subsets for which~$\boldw$ is good. For example, $S_1 = \{1, 9\}, S_2 = \{2\}, S_3 = \{3, 5, 6\}, S_4 = \{4, 7\}$, and~$S_5 = \{8\}$. Then, we have $\{w_s\vert s\in S_1 \}\subseteq F_1$, $\{w_s\vert s\in S_2 \}\subseteq F_1$, $\{w_s\vert s\in S_3 \}\subseteq F_4$, $\{w_s\vert s\in S_4 \}\subseteq F_3$, and $\{w_s\vert s\in S_5 \}\subseteq F_2$.
\end{example}

\subsection{Our protocol}\label{section:jointprotocol}
We now present a protocol that is a generalization of the coset protocol in Section \ref{section:pmonealgos}. Recall that~$t \in [n]$ is a privacy parameter of the protocol, and fix~$q$ as any non-negative integer power of two such that~$q\le \min\{t,2^{m-1}\}$. Similar to Section \ref{section:pmonealgos}, for a given partition $\mathcal{S} = \{S_1, \ldots, S_t\}$ of $[n]$, let $V_\cS \triangleq \operatorname{Span}_{\bF_{2^m}}\{\1_{S_i}\}_{i=1}^t$, where $\1_{S_i}$ the characteristic vector of~$S_i$ over~$\bF_{2^m}$. 
In addition, let $M \in \bF_{2^m} ^{(n-t)\times n}$ be a parity-check matrix for $V_\cS$, i.e.,~$V_\cS$ is the right kernel of~$M$ over~$\bF_{2^m}$, and let $B(M, \boldb)$ be some deterministic algorithm that finds a solution~$\boldy$ to the equation $M \odot \boldy^\intercal = \boldb$ over $\bF_{2^m}$. The algorithm~$B$ is assumed to be publicly known. The protocol works as follows.

\begin{enumerate}[label=\roman*)]
    \item The server:
    \begin{enumerate}[label=\alph*)]
    \item Finds and publishes a partition~$\cS= \{S_1, \ldots, S_t\}$ of $[n]$ for which~$\boldw$ is good.
    \item Publishes $\boldq \triangleq M \odot \boldw^\intercal$.
    \item Defines $\boldu \triangleq B(M, \boldq)$ and keeps it private.
    \item Finds a matrix $U_i\in \{\pm1\}^{m\times |S_i|}$ that is the $\{\pm1\}$-matrix representation of~$\boldu\vert_{S_i}$ for every $i \in [t]$. 
    \item Finds the unique $\ell_1, \ldots, \ell_t \in \bF_{2^m}$ such that $$\boldu = \boldw \oplus \bigoplus_{i=1}^t(\ell_i \odot \1_{S_i})$$ and keeps them private. These $\ell_i$'s exist since $M \odot \boldu^\intercal = M \odot \boldw^\intercal$, and thus $$\boldu \in V_\cS \oplus \boldw = \operatorname{Span}_{\bF_{2^m}}\{\1_{S_i}\}_{i=1}^t \oplus \boldw.$$ 
    \end{enumerate}
    \item The user:
    \begin{enumerate}[label=\alph*)]
    \item Defines $\boldu \triangleq B(M, \boldq)$. This is the same $\boldu$ as the one defined by the server since $B$ is deterministic.
    \item Finds a matrix $U_i\in \{\pm1\}^{m\times |S_i|}$ that is the~$\{\pm1\}$-matrix representation of~$\boldu\vert_{S_i}$ for every $i \in [t]$ (possible since~$\cS$ has been published). 
    \item Finds a submatrix~$R_i \in \{\pm 1\}^{r_i \times |S_i|}$ of~$U_i$ which contains the lexicographically first maximal set of linearly independent (over~$\bR$) rows of~$U_i$, with $r_i$ denoting the size of this set.
    \item Sends~$\{R_i(\boldx\vert_{S_i})^\intercal\}_{i=1}^t$ to the server (computed over~$\bR$).
    \end{enumerate}
    \item The server:
    \begin{enumerate}[label=\alph*)]
    \item Finds matrices $R_i \in \{\pm 1\}^{r_i \times |S_i|}$ for all $i \in [t]$ in the same manner that the user does. 
    \item Computes matrices $Q_i\in \bR^{m \times r_i}$ such that $Q_iR_i=U_i$ for all $i \in [t]$.
    \item 
    Computes and outputs $$\sum_{i=1}^{t}\diag(\ell_i)Q_iR_i(\boldx\vert_{S_i})^\intercal,$$ where $\diag(\ell_i)$ is the diagonal matrix with the entries of the vector in $\{\pm 1\}^m$ corresponding to $\ell_i$ on its diagonal.
    \end{enumerate}
\end{enumerate}

\subsection{Correctness of the protocol}

To prove the correctness of the above protocol, we first need the following lemma.  

\begin{lemma}\label{lemma:mminuslogq}
    Let $V_\cS \triangleq \operatorname{Span}_{\bF_{2^m}}\{\1_{S_i}\}_{i=1}^t$ for some partition $\mathcal{S} = \{S_1, \ldots, S_t\}$ of $[n]$, and let~$\cF=\{F_1,\ldots,F_q\}$ be the partition of~$\bF_{2^m}$ according to Lemma \ref{lemma:fieldpartition}. 
    For every~$\boldw\in\bF_{2^m}^n$ such that $\boldw$ is good for $\cS$,
    we have that~$$\dim(\RSpan ((V_\cS\oplus\boldw)\vert_{S_i}))= m - \log q$$ for all~$i\in[t]$.
\end{lemma}

\begin{proof}
    For every~$\boldu\in V_\cS \oplus \boldw$, we have that $$\boldu = \boldw \oplus \bigoplus_{i=1}^t(\ell_i \odot \1_{S_i})$$ for some unique $\ell_1, \ldots, \ell_t \in \bF_{2^m}$. Therefore, for all~$i\in[t]$, we have $$\boldu\vert_{S_i} = \boldw\vert_{S_i} \oplus \ell_i.$$ In other words, each entry of~$\boldu\vert_{S_i}$ is given by adding (in~$\bF_{2^m}$)~$\ell_i$ to the corresponding entry in~$\boldw\vert_{S_i}$. When referring to the respective~$\{\pm 1\}$-matrix representations~$U_i$ and $W_i\triangleq \boldW\vert_{S_i}$ of~$\boldu\vert_{S_i}$ and~$\boldw\vert_{S_i}$, respectively, this translates to
    \begin{align*}
        U_i=\diag(\ell_i) W_i.
    \end{align*}
    Therefore, since~$\diag(\ell_i)$ is an invertible matrix over~$\bR$, it follows that for every~$\boldu\in V_\cS\oplus \boldw$, $U_i$ has the same row span as~$W_i$. Hence, we have that~$$\RSpan((V_\cS\oplus \boldw)\vert_{S_i})=\RSpan(\boldw\vert_{S_i}).$$

    It remains to show that $\dim(\RSpan(\boldw\vert_{S_i}))=m-\log q$. Since the entries of $\boldw\vert_{S_i}$ all come from some subset $F_j$ of the field $\bF_{2^m}$ and are thus contained in an $(m - \log q)$-dimensional subspace over $\bR$, the claim follows.   
\end{proof}

Now we are ready to show that our protocol does indeed yield the desired outcome.

\begin{lemma}\label{lemma:algorithmcorrectness}
    The final result of the algorithm computed by the server is~$\boldW\boldx^\intercal$.
\end{lemma}
\begin{proof}
    From the proof of Lemma \ref{lemma:mminuslogq}, it follows that $U_i = \diag(\ell_i) W_i$ for all~$i \in [t]$. Recall that $\ell_1, \ldots, \ell_t \in \bF_{2^m}$. This means that $$\diag(\ell_i)Q_iR_i = \diag(\ell_i) U_i = W_i$$ for every $i$. Therefore, 
    \begin{align*}
        \sum_{i=1}^{t}\diag(\ell_i)Q_iR_i(\boldx\vert_{S_i})^\intercal &= \sum_{i=1}^{t} W_i(\boldx\vert_{S_i})^\intercal \\ &= \sum_{i=1}^{t} (\boldW\vert_{S_i})(\boldx\vert_{S_i})^\intercal = \boldW\boldx^\intercal. \qedhere
    \end{align*}
\end{proof}

\subsection{Merit of the protocol}

\begin{theorem}\label{theorem:jointprotocolmerit}
    For any $t \in [n]$ and any integer power of two $q\le \min\{t,2^{m-1}\}$, 
    the above joint retrieval protocol has user-privacy $\ell = t(m - \log q)$, server-privacy $I(W; Q) = m(n - t)$ and publication cost $d \le (n-t) \log t + nH(\frac{t}{n}) + m(n-t)$, where $H$ is the binary entropy function.
\end{theorem}
\begin{proof}
    The server needs to communicate the identity of the coset containing $\boldw$, in the form of the query $\boldq$, as well as the identity of the partition $\cS$. Publishing the query requires $m(n - t)$ bits. The number of partitions of $[n]$ into $t$ pairwise disjoint subsets equals the Stirling number of the second kind, denoted by $S(n, t)$, which is upper bounded by $\frac{1}{2}\binom{n}{t}t^{n-t}$~\cite{rennie1969stirling}. Thus,
    \begin{align*}
        d &= m(n - t) + \log S(n, t) \\ &\le m(n - t) + \log \frac{1}{2}\binom{n}{t}t^{n-t} \\ &\le m(n-t) + (n-t) \log t + nH(\frac{t}{n}).
    \end{align*}
    
    Recall that for every $i \in [t]$, the matrix~$R_i$ found by the user is of dimension $r_i \times |S_i|$, and from Lemma~\ref{lemma:mminuslogq} it follows that~$r_i = m - \log q$. Therefore, the dimension of the subspace on which~$\boldx$ is revealed is at most~$t(m - \log q)$; that is, user-privacy~$\ell = t(m - \log q)$. 

    It remains to show that $I(W; Q) = m(n - t)$. Since $$I(W; Q) = H(W) - H(W|Q) = mn - H(W|Q),$$ it suffices to show that $H(W|Q) = mt$. Notice that disclosing $Q = \boldq$ reveals the identity of the coset of $V_\cS$ to which $\boldw$ belongs. Since $\boldw$ is uniformly distributed over $\bF_{2^m}^n$, and because cosets of $V_\cS$ are of size $2^{mt}$, it follows that $W|Q = \boldq$ is uniform over a set of size $2^{mt}$, and therefore $H(W|Q = \boldq) = mt$. Hence, \begin{align*}
        H(W|Q) &= \sum_{\boldq}\Pr(\boldq)H(W|Q=\boldq) \\ &= H(W|Q = \boldq) = mt. \qedhere
    \end{align*}
\end{proof}

In Section~\ref{section:perfectptimality}, it is shown that the above joint retrieval protocol attains optimal privacy tradeoff for all~$t \ge 2^{m-1}$ by setting~$q=2^{m-1}$. However, the publication cost is not optimal.






\section{Perfect and hard set PI problems}\label{section:perfectandhard}


This section studies private inference problems where the values of the weight vector are taken from two special kinds of finite sets of real numbers, \textit{perfect sets} and \textit{hard sets}. In simple terms, a set of coefficients~$A$ is called \textit{perfect} if it allows for the most concise representation as a linear combination of~$\{\pm 1\}$ vectors. Conversely, it is termed \textit{hard} if any such representation includes a maximum number of~$\{\pm 1\}$ vectors. The sense in which such representation is minimal or maximal is captured by a newly defined notion called \textit{coefficient complexity}.

Coefficient complexity is an additive combinatorial notion of complexity of a finite subset of $\bR$, recently developed in \cite{ramkumar2023coefficient}. We will subsequently define what a \textit{dictionary} is for a finite set based on its coefficient complexity. Before doing so, we introduce some notations used in this section.

In this section, we denote $\1$ as the all $1$'s vector over $\bR$. For a finite set $A$, we denote $\operatorname{vec}(A)$ as the column vector of the elements of~$A$ in increasing order. We denote $\operatorname{sum}(A)$ as the sum of the elements of $A$. We use $\operatorname{colspan}_{\bR}(M)$ for denoting the column span of matrix~$M$ over $\bR$.

We use standard notation for the addition of sets. For any two sets $A, B \subseteq \bR$, $A+B \triangleq \{a+b\vert a\in A, b\in B\}$, and we use $\addsets_i A_i$ for addition of multiple sets $A_i$. Furthermore, we employ the notations $A+\gamma \triangleq \{a+\gamma\vert a\in A\}$ and $\gamma A \triangleq \{\gamma a\vert a\in A\}$ for $\gamma \in \bR$. 

\subsection{Coefficient complexity, dictionaries, and perfect and hard sets}

The authors of \cite{ramkumar2023coefficient} developed the definition of \textit{coefficient complexity} for a finite set as follows:

\begin{definition}\cite{ramkumar2023coefficient}\label{def:complexity}
    For any finite set $A \subset \bR$, the coefficient complexity of $A$, denoted by $C(A)$, is the smallest positive integer $\theta$ for which there exist $\theta + 1$ not necessarily distinct real numbers $\lambda_0, \lambda_1, \ldots, \lambda_{\theta}$ such that
    \begin{align*}
        A \subseteq \addsets_{i=1}^{\theta} \lambda_i \{1, -1\} + \lambda_0.
    \end{align*}
    
\end{definition}

\begin{example}\label{example:complexity}
    From the definition above, $C(\{-3, -1, 1, 3\}) = 2$ as $$\{-3, -1, 1, 3\} = -2\{1, -1\} + \{1, -1\}.$$ In comparison, consider $\{-2, 0, 1, 2\}$. There exist no three real coefficients $\lambda_0, \lambda_1, \lambda_2$ such that $\{-2, 0, 1, 2\} \subseteq \lambda_1\{1, -1\} + \lambda_2\{1, -1\} + \lambda_0$. On the other hand, $$\{-2, 0, 1, 2\} \subseteq -\frac{1}{4}\{1, -1\}+\frac{3}{4}\{1, -1\}+\frac{5}{4}\{1, -1\}+\frac{1}{4}.$$ Therefore, $C(\{-2, 0, 1, 2\}) = 3$.
\end{example}

This notion of complexity quantifies the extent to which the vector of coefficients can be expressed as a linear combination of vectors over~$\{\pm1\}$. To this end, we establish the notion of a \textit{dictionary} for a finite set $A$.

\begin{definition}\label{definition:dictionary}
    For any finite set $A \subset \bR$ with $C(A) = \theta$, a dictionary for $A$ is a matrix $D_A \in \{\pm 1\}^{|A| \times (\theta+1)}$ that contains a column of all $1$'s and satisfies $\operatorname{vec}(A) \in \operatorname{colspan}_{\bR}(D_A)$.  
\end{definition}
\begin{example}\label{example:dictionary}
    Consistent with Example~\ref{example:complexity}, the matrix $$\begin{bmatrix}
        1 & 1 & -1\\
        1 & 1 & 1\\
        1 & -1 & -1\\
        1 & -1 & 1
    \end{bmatrix}$$is a dictionary for the set~$\{-3,-1,1,3\}$, since
    \begin{align*}
    \begin{bmatrix}
        -3 \\
        -1 \\
        1 \\
        3 
    \end{bmatrix} = 0 \cdot \begin{bmatrix}
        1 \\
        1 \\
        1 \\
        1 
    \end{bmatrix} -2 \cdot \begin{bmatrix}
        1 \\
        1 \\
        -1 \\
        -1 
    \end{bmatrix} + 1 \cdot \begin{bmatrix}
        -1 \\
        1 \\
        -1 \\
        1 
    \end{bmatrix}.\end{align*}
    The matrix $$\begin{bmatrix}
        1 & 1 & -1 & -1\\
        1 & -1 & 1 & -1\\
        1 & -1 & -1 & 1\\
        1 & 1 & 1 & 1
    \end{bmatrix}$$is a  dictionary for the set~$\{-2,0,1,2\}$, since
    \begin{align*}
    \begin{bmatrix}
        -2 \\
        0 \\
        1 \\
        2 
    \end{bmatrix} = \frac{1}{4} \cdot \begin{bmatrix}
        1 \\
        1 \\
        1 \\
        1 
    \end{bmatrix} -\frac{1}{4} \cdot \begin{bmatrix}
        1 \\
        -1 \\
        -1 \\
        1 
    \end{bmatrix} + \frac{3}{4} \cdot \begin{bmatrix}
        -1 \\
        1 \\
        -1 \\
        1 
    \end{bmatrix} + \frac{5}{4} \cdot \begin{bmatrix}
        -1 \\
        -1 \\
        1 \\
        1 
    \end{bmatrix}.\end{align*}
    
\end{example}
In addition, \cite{ramkumar2023coefficient} provided the following upper and lower bounds for coefficient complexity. 

\begin{lemma}\cite{ramkumar2023coefficient}\label{lemma:complexitybounds}
    For any finite set $A$ of size at least two, 
    \begin{align*}
        \left\lceil\log |A|\right\rceil \le C(A) \le |A| - 1.
    \end{align*}
\end{lemma}

We will focus on sets of the lowest and highest coefficient complexities and formally define \textit{perfect and hard sets} as follows.

\begin{definition}
    Any set $A \subset \bR$ of size $2^m$, where $m$ is a positive integer, is called a perfect set if $C(A) = \log |A|$ and $\operatorname{sum}(A) = 0$.
\end{definition}

It is easy to check that~$\{-3, -1, 1, 3\}$ is a perfect set. Perfect sets are interesting because the~$m$-bit midriser quantizer~\cite{gersho1978principles}, which is commonplace in communication and signal processing systems, results in a perfect set of size~${2^m}$. Based on the above definition, the following equivalent characterization of a perfect set is immediate.

\begin{lemma}
    Every perfect set $A$ can be written as $A = \addsets_{i=1}^{\log |A|} \lambda_i \{1, -1\}$ for some real numbers $\lambda_1, \ldots, \lambda_{\log |A|}$.
\end{lemma}
\begin{proof}
    Because the set $A$ is perfect, $C(A) = \log |A|$. The equality follows from the fact that there are exactly $|A|$ elements in $\addsets_{i=1}^{\log |A|} \lambda_i \{1, -1\} + \lambda_0$. Finally, since the sum of elements in $ \addsets_{i=1}^{\log |A|} \lambda_i \{1, -1\}$ is zero, it follows that $\operatorname{sum}(A) = 0$ if and only if $\lambda_0 = 0$.
\end{proof}

\begin{definition}
     Any set $A \subset \bR$ of size $2^m$, where $m$ is a positive integer, is called a hard set if $C(A) = |A| - 1$.
\end{definition}
It immediately follows that if $A$ is a hard set, then $$A \subseteq \addsets_{i=1}^{|A| - 1} \lambda_i \{1, -1\} + \lambda_0$$ for some $\lambda_0, \lambda_1, \ldots, \lambda_{|A| - 1} \in \bR$. It is shown in \cite{ramkumar2023coefficient} that most sets are hard.

Next, we investigate PI problems where the entries of weight vector $\boldw$ belong to either a perfect set or a hard set. 

\subsection{The perfect set PI problem}\label{section:perfectset}

Every perfect set of size~$2^m$ can be written as $$A = \addsets_{i=1}^{m} \lambda_i \{1, -1\}.$$ It follows that any $\boldw \in A^n$ can be written as $$\boldw = \sum_{i=1}^{m} \lambda_i \boldw^{(i)},$$ where $\boldw^{(i)} \in \{\pm 1\}^n$ for every~$i$. Since~$\boldw$ is randomly chosen from a uniform distribution $W = \textrm{Unif}(A^n)$, the resulting distributions $\boldw^{(1)}, \ldots, \boldw^{(m)}$ are mutually independent and each is uniform over~$\{\pm1\}^n$.
This follows from the observation that if $D_A$ is a dictionary for perfect set $A$, then the $2^m \times m$ matrix obtained by removing the all $1$'s column from $D_A$ contains all possible $\{\pm 1\}^m$ vectors as rows; this enables one to run~$m$ independent instances of retrieving one signal without any loss of privacy.

On that note, we now demonstrate that the joint retrieval approach in Section \ref{section:jointprotocol} leads to a protocol for this problem. The joint retrieval protocol can be used to obtain $m$ inner products with the data at the same time. After obtaining the $m$ individual signals $\{\boldw^{(i)}\boldx^\intercal\}_{i=1}^m$, the server linearly combines them using the coefficients $\{\lambda_i\}_{i=1}^{m}$ to compute $$\boldw\boldx^\intercal = \sum_{i=1}^{m} \lambda_i \boldw^{(i)}\boldx^\intercal.$$ For any $t \in [n]$ and any integer power of two $q\le \min\{t,2^{m-1}\}$, this results in publication cost $d \le (n-t) \log t + nH(\frac{t}{n}) + m(n-t)$, user-privacy $\ell = t(m - \log q)$, and server-privacy $I(W; Q) = m(n - t)$, due to Theorem \ref{theorem:jointprotocolmerit}. In Section~\ref{section:perfectptimality}, we show that for $t \ge 2^{m-1}$, this perfect set protocol with $q=2^{m-1}$ is optimal in terms of the privacy tradeoff.

\subsection{The hard set PI problem}
We now turn to hard sets: $$A \subseteq \addsets_{i=1}^{2^m - 1} \lambda_i \{1, -1\} + \lambda_0,$$ with $|A| = 2^m$ for some positive integer $m$, and the server holds $\boldw \in A^n$, randomly chosen from a uniform distribution $W = \textrm{Unif}(A^n)$. Any $\boldw \in A^n$ can be written as $$\boldw = \sum_{i=1}^{2^{m}-1} \lambda_i \boldw^{(i)} + \lambda_0 \1.$$ Here, $\boldw^{(1)}, \ldots, \boldw^{(2^{m}-1)}$ are not necessarily uniform and independent. As a result, we cannot directly apply the joint retrieval approach to this problem.

In what follows, we employ Hadamard matrices, specifically those under Sylvester's construction \cite{horadam2012hadamard}. We recursively construct Hadamard matrices of order $2^m$ for every non-negative integer~$m$ as follows, with base case~$H_1=[1]$:
\begin{align*}
    H_{2^m} &= \begin{bmatrix}
        H_{2^{m-1}} & H_{2^{m-1}} \\
        H_{2^{m-1}} & -H_{2^{m-1}}
    \end{bmatrix}.
\end{align*}
\subsubsection{Technical lemmas}
Before introducing the protocol for the hard set PI problem, we remind the reader that~$H_{2^m}H_{2^m}^\intercal=2^mI$ and~$H_{2^m}=H_{2^m}^\intercal$ for all~$m$. 
We will also show among other results about Hadamard matrices that they are valid candidates for dictionaries for hard sets.

\begin{lemma}\label{lemma:hadamardclaims}
    For any positive integer $m$, we have the following:
    \begin{enumerate}[label=(\roman*)]
        \item The matrix~$H_{2^m}$ is a dictionary for any hard set of size~$2^m$.
        \item There exists a~$2^m \times m$ submatrix of~$H_{2^m}$ whose rows are precisely all vectors in~$\{\pm 1\}^m$.
        \item The element-wise product of any two (not necessarily distinct) rows of~$H_{2^m}$ is a row of the same matrix.
    \end{enumerate}
\end{lemma}
\begin{proof}
    To prove the first claim, let~$$\mathbf{\lambda} = (\lambda_0, \lambda_1, \ldots, \lambda_{2^m-1}) \triangleq \frac{1}{2^m}\operatorname{vec}(A)^\intercal H_{2^m}.$$ Then, since $H_{2^m}H_{2^m}^\intercal = 2^m I$, it follows that $$\operatorname{vec}(A) = H_{2^m}\mathbf{\lambda}^\intercal,$$ indicating that~$\operatorname{vec}(A)$ is in the column span of~$H_{2^m}$ over~$\bR$ and that~$H_{2^m}$ is a valid dictionary by Definition~\ref{definition:dictionary}. The last two claims follow from the fact~\cite{macwilliams1977theory} that the rows of~$H_{2^m}$ are an~$\bF_2$-linear code of dimension~$m$, and a full proof is given in Appendix~\ref{section:omitted}. 
\end{proof}
\begin{remark}
    The second claim in Lemma~\ref{lemma:hadamardclaims} indicates that for the Hadamard matrix of order~$2^m$, there exist~$m$ columns upon which the other~$2^m-m$ columns depend. In other words, for every row in the matrix, knowing the entries of these $m$ columns suffices to identify the entire row. The proof of this claim, given in Appendix~\ref{section:omitted}, provides an explicit way to select these~$m$ columns. Note that the first column of~$H_{2^m}$ is a column of all $1$'s and thus cannot be one of the chosen columns. 
\end{remark}

We are now ready to present a protocol for the hard set PI problem which employs Hadamard matrices. This protocol is a generalization of the improved random key protocol in Section \ref{section:pmonealgos}.

\subsubsection{The protocol}\label{section:hadamardalgo}
In the sequel, we assume that the hard set $A$ of size $2^m$ is public knowledge, and that the user and the server agree on the dictionary $D_A = H_{2^m}$, as well as the vector of coefficients $\mathbf{\lambda} = (\lambda_0, \lambda_1, \ldots, \lambda_{2^m-1}) = \frac{1}{2^m}\operatorname{vec}(A)^\intercal H_{2^m}$. We also assume that they agree on the same~$m$ columns of $H_{2^m}$ identified in Lemma~\ref{lemma:hadamardclaims}.$(ii)$, and let~$L\subseteq[2^m]$ be their indices with~$|L| = m$. Observe that $\boldw^{(i)}$ depends on the $(i+1)$-th column of $H_{2^m}$. To handle this discrepancy in indexing, we define $\Tilde{L} \triangleq L - 1 = \{j-1\vert j \in L\}$.
The partition $\cS = \{S_1, \ldots, S_t\}$ of $[n]$ is also publicly known. Our protocol works as follows:
\begin{enumerate}[label=\roman*)]
    \item The server:
    \begin{enumerate}[label=\alph*)]
    \item Computes $\boldw^{(1)}, \ldots, \boldw^{(2^m-1)} \in \{\pm 1\}^n$ such that $\boldw = \sum_{i=1}^{2^m-1}\lambda_i \boldw^{(i)} + \lambda_0 \1$.
    \item Defines~$\hat{\boldw}^{(i)} \in \{\pm 1\}^n$ for every~$i \in [2^m-1]$ as follows: $$\hat{\boldw}^{(i)}\vert_{S_j} \triangleq w_{S_{j,1}}^{(i)}\cdot \boldw^{(i)}\vert_{S_j}$$ for every~$j \in [t]$. 
    \item Defines the last~$|S_j|-1$ entries of~$\hat{\boldw}^{(i)}\vert_{S_j}$ as~$$\boldu_{i,j} \triangleq \hat{\boldw}^{(i)}\vert_{S_j\setminus\{S_{j,1}\}},$$ for every~$i \in [2^m-1]$ and~$j \in [t]$.
    \item Publishes $\boldu_{i,j}$ for every~$i \in \Tilde{L}$ and~$j \in [t]$.
    \end{enumerate}
    \item The user:
    \begin{enumerate}[label=\alph*)]
    \item Obtains~$\hat{\boldw}^{(i)}$ for every~$i \in [2^m-1]$. This can be done since the user knows~$\boldu_{i,j}$ for every~$i \in \Tilde{L}$ and~$j \in [t]$, as well as the dictionary~$D_A$; see the correctness proof below for a detailed argument.
    \item Computes and sends to the server $(\hat{\boldw}^{(i)}\vert_{S_j})(\boldx\vert_{S_j})^\intercal$ for every $i \in [2^m-1]$ and~$j \in [t]$, as well as $\1 \boldx^\intercal$.
    \end{enumerate}
    \item The server:
    \begin{enumerate}[label=\alph*)]
    \item Computes $$(\boldw^{(i)}\vert_{S_j})(\boldx\vert_{S_j})^\intercal = w_{S_{i,1}}^{(j)}\cdot(\hat{\boldw}^{(i)}\vert_{S_j})(\boldx\vert_{S_j})^\intercal$$ for every~$i \in [2^m-1]$ and~$j \in [t]$.
    \item Computes $$\sum_{i=1}^{2^m-1} \sum_{j=1}^t \lambda_i (\boldw^{(i)}\vert_{S_j})(\boldx\vert_{S_j})^\intercal + \lambda_0 \1 \boldx^\intercal = \boldw\boldx^\intercal.$$
    \end{enumerate}
\end{enumerate}

Next, we show the correctness of this protocol by arguing that the user can perform Step ii.a).

\subsubsection{Correctness of the protocol}\label{section:correctnessofhardalgo}

The user receives~$\boldu_{i,j}$ for every~$i \in \Tilde{L}$ and~$j \in [t]$ from the server. Therefore, the user can obtain~$\hat{\boldw}^{(i)}$ for every~$i \in \Tilde{L}$, since $\hat{\boldw}^{(i)}\vert_{S_j} = (1, \boldu_{i,j})$. We now show that using $\hat{\boldw}^{(i)}$ for all $i \in \Tilde{L}$ and the dictionary $D_A = H_{2^m}$, the user is able to compute $\hat{\boldw}^{(i)}$ for all $i \in [2^m-1] \setminus \Tilde{L}$ as well. Define a matrix $\boldW \in \{\pm 1\}^{2^m \times n}$ as follows:
\begin{align*}
    \boldW = \begin{bmatrix}
        \1 \\
        \boldw^{(1)} \\
        \vdots \\
        \boldw^{(2^m-1)}
    \end{bmatrix}.
\end{align*}
It can be seen that the transpose of every column of $\boldW$ is a row of $D_A = H_{2^m}$.


We also define a different matrix $\hat{\boldW} \in \{\pm 1\}^{2^m \times n}$ similarly as follows:
\begin{align*}
    \hat{\boldW} = \begin{bmatrix}
        \1 \\
        \hat{\boldw}^{(1)} \\
        \vdots \\
        \hat{\boldw}^{(2^m-1)}
    \end{bmatrix}.
\end{align*}
Fix some $S_j\in \cS$. From the definition of $\hat{\boldw}^{(i)}$, the columns of $\hat{\boldW}\vert_{S_j}$ are precisely the respective element-wise product of every column of $\boldW\vert_{S_j}$ with the column $\boldW\vert_{S_{j,1}}$. Since every column of $\boldW$ is the transpose of some row of $D_A$, according to Lemma \ref{lemma:hadamardclaims}.$(iii)$, every column of $\hat{\boldW}$ is also the transpose of some (possibly different) row of $D_A$. Hence, by Lemma \ref{lemma:hadamardclaims}.$(ii)$, if the user knows the dictionary $D_A$ as well as the rows of $\hat{\boldW}$ indexed by $L$, it can construct the entire $\hat{\boldW}$. 

\subsubsection{An illustrative example}
We now illustrate the protocol with the following example.
\begin{example}
    Let $m = 3$ and let $A = \{a_1, \ldots, a_8\}$ be a hard set of size $8$, with $a_1 < \cdots < a_8$. Each $a_i$ corresponds to row $i$ of the dictionary
    \begin{align*}
        D_A = H_8 = \begin{bmatrix}
        1 & \phantom{-}1 & \phantom{-}1 & \phantom{-}1 & \phantom{-}\blue{1} & \phantom{-}1 & \phantom{-}\blue{1} & \phantom{-}\blue{1}\\
        1 & -1 & \phantom{-}1 & -1 & \phantom{-}\blue{1} & -1 & \phantom{-}\blue{1} & \blue{-1}\\ 1 & \phantom{-}1 & -1 & -1 & \phantom{-}\blue{1} & \phantom{-}1 & \blue{-1} & \blue{-1}\\ 1 & -1 & -1 & \phantom{-}1 & \phantom{-}\blue{1} & -1 & \blue{-1} & \phantom{-}\blue{1}\\ 1 & \phantom{-}1 & \phantom{-}1 & \phantom{-}1 & \blue{-1} & -1 & \blue{-1} & \blue{-1}\\ 1 & -1 & \phantom{-}1 & -1 & \blue{-1} & \phantom{-}1 & \blue{-1} & \phantom{-}\blue{1}\\1 & \phantom{-}1 & -1 & -1 & \blue{-1} & -1 & \phantom{-}\blue{1} & \phantom{-}\blue{1}\\ 1 & -1 & -1 & \phantom{-}1 & \blue{-1} & \phantom{-}1 & \phantom{-}\blue{1} & \blue{-1}
    \end{bmatrix}.
    \end{align*}
    Fix~$L = \{5, 7, 8\}$ (corresponding columns colored in blue) as an index set that satisfies the requirement of Lemma~\ref{lemma:hadamardclaims}.(ii). This results in $\Tilde{L} = \{4,6,7\}$. Now, say $\boldw = [a_2, a_3, a_4, a_5]$. As noted above, $\boldw^{(i)}$ depends on the $(i+1)$-th column of $H_{8}$. Hence, by looking at rows $2,3,4$ and $5$ of the colored columns we get $$\boldw^{(4)} = [\phantom{-}1, \phantom{-}1, \phantom{-}1,-1],$$ $$\boldw^{(6)} = [\phantom{-}1, -1, -1,-1],$$ $$\boldw^{(7)} = [-1, -1, \phantom{-}1,-1].$$ Now, suppose $t = 1$, which means $S_1 = [4]$. Then we have 
    $$\hat{\boldw}^{(4)} = w^{(4)}_1 \cdot \boldw^{(4)} = 1 \cdot \boldw^{(4)} =[\phantom{-}1, \phantom{-}1, \phantom{-}1,-1],$$ and similarly,
    \begin{align*}
        \hat{\boldw}^{(6)} &= 1 \cdot \boldw^{(6)} = [\phantom{-}1, -1, -1,-1]\\
        \hat{\boldw}^{(7)} &= -1 \cdot \boldw^{(7)} =[\phantom{-}1, \phantom{-}1, -1,\phantom{-}1].
    \end{align*}
   The server publishes $[1, 1, -1]$, $[-1, -1, -1]$ and $[1, -1, 1]$.
    
   Observe that~$\hat{w}_1^{(i)} = 1$ for all~$i \in [7]$. Clearly, the user is able to obtain~$\hat{\boldw}^{(4)}$, $\hat{\boldw}^{(6)}$ and $\hat{\boldw}^{(7)} $ by prefixing ``$1$'' to the corresponding queries. Now we focus on $\hat{\boldw}^{(5)}$ as an example to illustrate that the user can obtain $\hat{\boldw}^{(i)}$ for all~$i \in [7] \setminus \Tilde{L}$. To obtain $\hat{w}^{(5)}_2$, the user inspects $D_A$ for the row containing $1, -1, 1$ in the $5$th, $7$th and $8$th entries, respectively; to obtain $\hat{w}^{(5)}_3$, the row containing $1, -1, -1$ in the $5$th, $7$th and $8$th entries, respectively, and to obtain $\hat{w}^{(5)}_4$, the row containing $-1, -1, 1$ in the $5$th, $7$th and $8$th entries, respectively. As a result, $$\hat{\boldw}^{(5)}\vert_{\{2,3,4\}} = [-1, 1, 1],$$ and $$\hat{\boldw}^{(5)} = [1, -1, 1, 1].$$ Also, from the dictionary the server knows that $$\boldw^{(5)} = [-1, 1, -1, -1].$$ Thus, the server is able to retrieve $\boldw^{(5)}\boldx^\intercal$ based on $\hat{\boldw}^{(5)}\boldx^\intercal$ from the user (i.e., $\boldw^{(5)}\boldx^\intercal = -\hat{\boldw}^{(5)}\boldx^\intercal$). Similarly, the server can retrieve $\boldw^{(i)}\boldx^\intercal$ for all $i \in [7]$ and linearly combine them to retrieve $\boldw\boldx^\intercal$.
\end{example}
\subsubsection{Merit of the protocol}

We now present a theorem that provides the figures of merit of the above protocol.
\begin{theorem}
    For any privacy parameter~$t \in [n]$, the above protocol for the hard set PI problem has publication cost $d = m(n - t)$, server-privacy $I(W; Q) = m(n - t)$, and user-privacy $\ell = (2^m-1)t+1$. 
\end{theorem}

\begin{proof}
    The user-privacy claim is immediate by definition. The server saves one bit of publication cost in every~$\hat{\boldw}^{(i)}\vert_{S_j}$ for every $i \in \Tilde{L}$ with $|\Tilde{L}| = m$ and for every~$j \in [t]$, thus leading to a total publication cost of $m(n - t)$ bits. To compute $I(W; Q) = H(W) - H(W|Q)$, first recall that $W$ is uniform over an alphabet of size $2^{mn}$, and so $H(W) = mn$. After the query is published, from the user's perspective, the only source of uncertainty of $\boldw$ comes from the keys $w_{S_{j,1}}^{(i)}$, and there are a total of $2^{mt}$ possibilities for their values. All these possibilities are equiprobable, since the $\boldw^{(i)}$'s for all $i \in \Tilde{L}$ are mutually independent and uniform over $\{\pm 1\}^n$. It follows that $W|Q = \boldq$ is uniform over a set of size $2^{mt}$, and $H(W|Q) = mt$.
\end{proof}

\begin{remark}
    For any arbitrary set $A$, let $\gamma_A \triangleq |\{i|\lambda_i\neq 0\}|$, where the $\lambda_i$'s are as defined in the proof of the first claim in Lemma \ref{lemma:hadamardclaims}. It follows from definition that if $A$ is a hard set of size $2^m$, then $\lambda_i \neq 0$ for all $i \in [2^m-1]$, resulting in $\gamma_A = 2^m-1$. If $A$ is not a hard set, then $\gamma_A$ can be smaller than $2^m-1$. The protocol given above for the hard set PI problem can be extended to the PI problem where the weight vector admits entries restricted to any arbitrary set $A$ of size $2^m$. For the resultant protocol, we have publication cost $d = m(n-t)$, server-privacy $I(W; Q) = m(n-t)$, and user-privacy $\ell = \gamma_A t+1$.
\end{remark}

\section{Bounds and optimality of protocols}\label{section:optimality}

As previously stated, naturally the quantities~$d$, $I(W; Q)$ and $\ell$ should be simultaneously minimized. In this section, however, we show that both privacy metrics cannot be minimized at the same time by deriving an inherent tradeoff. We also demonstrate a simple lower bound on publication cost. Then, we analyze the optimality of some of the protocols presented previously in terms of the privacy tradeoff and the bound on publication cost.

\subsection{General privacy tradeoff}

We show that for any finite set $A$, there is an inherent tradeoff between the two privacy measures. To this end, we need the following lemma, the proof of which is given in Appendix~\ref{section:omitted}.

\begin{lemma}\label{lemma:pmvectorsinsubspace}\cite{groenland2020intersection}
    For an $\bR$-subspace $S$ of dimension $\ell$, we have that $|S \cap A^n| \le |A|^\ell$ for any finite set $A$.
\end{lemma}

To derive the tradeoff, we assume polynomial decoding at the server's side. That is, there exists a polynomial $f_Q:\bR^\ell\to\bR$, which only depends on $Q$, such that $\boldw\boldx^\intercal = f_Q(\boldv_1\boldx^\intercal, \ldots, \boldv_\ell\boldx^\intercal)$ for every $\boldx$. 

\begin{theorem}\label{theroem:generallowerbound}
    For any finite set $A$, it holds that $I(W; Q)+\ell \cdot \log |A| \ge n \cdot \log |A|$.
\end{theorem}
\begin{proof}
    It is readily seen that $H(W) = n \cdot \log |A|$, so $$I(W; Q) = H(W) - H(W|Q) = n\cdot \log |A| - H(W|Q).$$ It then suffices to show that $H(W|Q) \le \ell \cdot \log |A|$. First, observe that since the vectors~$\boldv_i$ are a deterministic function of~$Q$, it follows that~$H(W\vert Q)=H(W,\{ \boldv_i \}_{i=1}^\ell\vert Q)$. Second, we assume polynomial decoding, i.e., that there exists~$f_Q:\bR^\ell\to\bR$ such that~$\boldw\boldx^\intercal = f_Q(\boldv_1\boldx^\intercal, \ldots, \boldv_\ell\boldx^\intercal)$. Since the scheme must be valid for any data distribution, it must be valid for every~$\boldx$. That is, the polynomial~$$f_Q(\boldv_1\boldx^\intercal, \ldots, \boldv_\ell\boldx^\intercal)-\boldw\boldx^\intercal,$$ seen as a polynomial in the~$n$ variables~$x_1,\ldots,x_n$, must be the zero polynomial. 
    
    Denote~$$f_Q(y_1,\ldots,y_\ell)=\sum_{\boldd\in \bN^\ell}f_\boldd\boldy^\boldd,$$ where~$$\boldy^\boldd=y_1^{d_1}\cdots y_\ell^{d_\ell},$$ and~$f_\boldd\in\bR$ for every~$\boldd$. It follows that for each~$i\in[n]$, the coefficient of~$x_i$ in~$$f_Q(\boldv_1\boldx^\intercal, \ldots, \boldv_\ell\boldx^\intercal)-\boldw\boldx^\intercal$$ is~$\sum_{r=1}^{\ell}f_{\bolde_r} v_{r,i}-w_i$, where~$\bolde_i$ is the~$i$-th unit vector of length~$\ell$. Setting these coefficients to zero yields the linear equation~$$\sum_{r=1}^\ell f_{\bolde_{r}}\boldv_r=\boldw,$$ and therefore~$$\boldw\in \RSpan\{ \boldv_i\}_{i=1}^\ell.$$
    
	
	Now, for query~$\boldq$ in the support of~$Q$, we bound the support size of the random variable~$W,\{ \boldv_i \}_{i=1}^\ell\vert Q=\boldq$ from above. Since any~$\boldw$ in this support is in the $\bR$-span of the vectors~$\boldv_1,\ldots,\boldv_\ell$, it follows that the size of this support cannot be larger than the maximum possible number of~$A^n$ vectors in an~$\bR$-subspace of dimension~$\ell$. Formally, it is readily verified that
	\begin{align*}
		|\operatorname{Supp}(W,\{ \boldv_i \}_{i=1}^\ell\vert Q=\boldq)|\le \max_{S\vert \dim_\bR(S)=\ell}|S\cap A^n|.
	\end{align*}
	Lemma~\ref{lemma:pmvectorsinsubspace} implies that $$\max_{S\vert \dim_\bR(S)=\ell}|S\cap A^n|\le |A|^\ell,$$ and hence (by \cite[Thm.~2.6.4]{CovThom06}) $H(W,\{ \boldv_i \}_{i=1}^\ell\vert Q=\boldq)$ cannot be larger than~$H(\textrm{Unif}(A^\ell))$, which is equal to~$\ell \cdot \log |A|$. Therefore,
	\begin{align*}
		H(W\vert Q)&=H(W,\{ \boldv_i \}_{i=1}^\ell\vert Q)\\&=\sum_{\boldq\in\operatorname{Supp}(Q)}\Pr(Q=\boldq) H(W,\{ \boldv_i \}_{i=1}^\ell\vert Q=\boldq)\\&\le \ell\cdot \log |A| \sum_{\boldq\in\operatorname{Supp}(Q) }\Pr(Q=\boldq)=\ell \cdot \log |A|. \qedhere
	\end{align*}
\end{proof}
Next, we provide a lower bound on the publication cost $d$. As mentioned in the problem definition, the number $d$ of published bits is a random variable that depends on $Q$. Therefore, the problem of minimizing the expected value of $d$ is a source coding problem, and the expected value of $d$ is lower bounded by $H(Q) = I(Q; W) + H(Q|W) = I(Q; W)$, leading to the following theorem.

\begin{theorem}\label{theorem:publicationcostbound}
    The expected value of $d$ is lower bounded by $I(W; Q)$.
\end{theorem}

\subsection{Optimality of the $\{\pm 1\}$ protocols}\label{section:pmoneoptimality}
For the $\{\pm 1\}$ PI problem, the general tradeoff in Theorem \ref{theroem:generallowerbound} specializes to the following:
\begin{align*}
    I(W; Q) + \ell \ge n.
\end{align*}
Because both protocols in Section \ref{section:pmonealgos} have $\ell = t$ and $I(W; Q) = n - t$, they indeed attain this tradeoff exactly and are therefore optimal in this sense. In addition, $d = n - t$ in both protocols, indicating that the publication cost is also optimal.




\subsection{Optimality of the perfect set protocol} \label{section:perfectptimality}

For a perfect set $A$ of size $2^m$, the general tradeoff specializes to the following:
\begin{align*}
    I(W; Q) + m\ell \ge mn.
\end{align*}
For any $t \in [n]$ and any integer power of two $q\le \min\{t,2^{m-1}\}$, the protocol presented in Section \ref{section:perfectset} for perfect sets of size $2^m$ has publication cost $d \le (n-t) \log t + nH(\frac{t}{n}) + m(n-t)$, user-privacy $\ell = t(m - \log q)$, and server-privacy $I(W; Q) = m(n - t)$. If we assume $t \ge 2^{m-1}$ and choose $q = 2^{m-1}$, then our scheme gives $\ell = t$ and $I(W; Q) = m(n - t)$. This tells us that the tradeoff is achieved in this protocol, provided $t \ge 2^{m-1}$. The publication cost, however, is not optimal.

\subsection{Optimality of the joint retrieval protocol} \label{section:jointoptimality}

The general tradeoff in Theorem \ref{theroem:generallowerbound} concerns the retrieval of one signal for a finite alphabet. That being said, because the joint retrieval approach works for the perfect set problem, the lower bound should also hold for the joint retrieval PI problem in Section \ref{section:jointretrieval}. Thus, we know that the joint retrieval protocol in Section \ref{section:jointprotocol} is also optimal in terms of the privacy tradeoff for $t \ge 2^{m-1}$, but not optimal in terms of publication cost. 
\section{Conclusion and future work}\label{section:conclusion}

We take an information-theoretic approach to private inference and view PI as the task of retrieving inner products of parameter vectors with data vectors, a fundamental operation in many machine learning models. We introduce protocols that enable retrieval of inner products for a variety of quantized models while guaranteeing some level of privacy for both the server and the user. Additionally, we derive a general inherent tradeoff between user and server privacy and show that our protocols are optimal in certain interesting settings. As of now, we do not have an optimal scheme for the hard set problem, or the $\{0, \pm 1\}$ problem in the Appendix. Closing the gap in these variants is left as future work. Furthermore, with the recently introduced notion of coefficient complexity \cite{ramkumar2023coefficient}, it remains to be explored if a more rigid privacy tradeoff which depends on coefficient complexity can be derived. 

It was recently shown in \cite{deng2023approximate} that with the $\{\pm 1\}$ problem, if an \textit{approximate} computation of the signal is sufficient for inference, then one can break the barrier presented in Section \ref{section:pmoneoptimality} and provide better privacy guarantees with the same publication cost. An approximation protocol that yields an approximate signal that is very close to the exact signal is also provided in \cite{deng2023approximate}. Extending the idea of approximating signals to cases other than $\{\pm 1\}$ to guarantee better server and user privacy is an intriguing future research direction.
\printbibliography

\appendices
\section{PI problem for roots of unity}\label{section:rootsofunity}

Here we look at a generalization of the $\{\pm 1\}$ case where the weights are restricted to roots of unity of any order. Denote by $R_p$ the set of roots of unity of order $p$, where $p$ is a positive integer. The server now holds a weight vector $\boldw \in R_p^n$ randomly chosen from a uniform distribution $W = \textrm{Unif}(R_p^n)$, and a user holds a complex-valued data vector $\boldx \in \bC^n$ that is randomly chosen from some continuous data distribution~$X$. Since $\boldw$ and $\boldx$ are both complex here, we measure the user-privacy $\ell$ in terms of complex projections as well.

\subsection{The roots-of-unity protocol}
\label{section:rootsofunityprotocol}
Our protocol for this problem variant works in a similar way to that in the improved random key protocol in Section \ref{section:pmonealgos}. Specifically, let $\ell_i \triangleq (w_{S_{i,1}})^{p-1}$ for every $i \in [t]$. Since the entries of $\boldw$ are all roots of unity of order $p$, it follows that $(w_{S_{i,1}})^{p-1}\cdot(w_{S_{i,1}}) = 1$, and thus the first entry of $(w_{S_{i,1}})^{p-1}\cdot \boldw\vert_{S_i}$ is $1$ for every $i$. The scheme works as follows:

\begin{enumerate}[label=\roman*)]
    \item The server publishes $(w_{S_{i,1}})^{p-1}\cdot \boldw\vert_{S_i \setminus \{S_{i,1}\}}$ for every $i \in [t]$.
    \item The user sends $$\{(1, (w_{S_{i,1}})^{p-1}\cdot \boldw\vert_{S_i \setminus \{S_{i,1}\}})(\boldx\vert_{S_i})^\intercal\}_{i=1}^{t}$$ to the server (computed over~$\bC$).
    \item The server computes and outputs $$\sum_{i=1}^{t}(w_{S_{i,1}})^{p-1}\cdot(1, (w_{S_{i,1}})^{p-1}\cdot \boldw\vert_{S_i \setminus \{S_{i,1}\}})(\boldx\vert_{S_i})^\intercal = \boldw\boldx^\intercal.$$   
\end{enumerate}

\begin{theorem}\label{theorem:rootofunityprotocolmerit}
    For any $t \in [n]$, the above protocol has publication cost $d = (n - t)\log p$, server-privacy $I(W; Q) = (n - t)\log p$, and user-privacy $\ell = t$ (complex projections).
\end{theorem}
\begin{proof}
    Publication cost and user-privacy results follow immediately from the protocol. It remains to show that $I(W; Q) = (n - t)\log p$. Since $$I(W; Q) = H(W) - H(W|Q) = n\log p - H(W|Q),$$ it suffices to show that $H(W|Q) = t\log p$. Notice that upon disclosing the query $\boldq$, revealing the key $\ell_i$ deterministically reveals the identity of $\boldw\vert_{S_i}$. Because each $\ell_i$ is chosen from $R_p$ of size $p$, it follows that $W|Q = \boldq$ is uniform over a set of size $p^t$, and therefore $H(W|Q = \boldq) = t\log p$. Hence,
    \begin{align*}
        H(W|Q) &= \displaystyle \sum_{\boldq} \Pr(\boldq)H(W|Q = \boldq) \\ &= t\log p \cdot \sum_{\boldq} \Pr(\boldq) \\ &= t\log p. \qedhere
    \end{align*}
\end{proof}

\subsection{Application to the~$\{0, \pm 1\}$ PI problem}

The above protocol for $R_3$ can be applied to the PI problem where $\boldw \in \{0, \pm 1\}^n$ and $\boldx \in \bR^n$. Now we measure $\ell$ in terms of real projections as $\boldw$ and $\boldx$ are real-valued in this case. Recall that~$R_3 = \{1, -\frac{1}{2}\pm \frac{\sqrt{3}}{2}i\}$. Given some $\boldw \in \{0, \pm 1\}^n$, the server defines a vector $\hat{\boldw} \in R_3^n$ as follows:
\begin{align*}
    \hat{w}_i = \begin{cases}
        1 & w_i = 0, \\ -\frac{1}{2}+ \frac{\sqrt{3}}{2}i & w_i = 1, \\-\frac{1}{2}- \frac{\sqrt{3}}{2}i & w_i = -1.
    \end{cases}
\end{align*}
We note that multiplying the imaginary part of $\hat{\boldw}$ by $\frac{2}{\sqrt{3}}$ gives $\boldw$. Then, we proceed with the roots-of-unity protocol with $\hat{\boldw}$. Finally, the server retrieves the desired inner product by considering the imaginary part of the outcome of the roots-of-unity protocol. It works as such:

\begin{enumerate}[label=\roman*)]
    \item The server finds $\hat{\boldw} \in R_3^n$ corresponding to $\boldw$ and publishes $(\hat{w}_{S_{i,1}})^{2}\cdot \hat{\boldw}\vert_{S_i \setminus \{S_{i,1}\}}$ for every $i \in [t]$.
    \item The user sends $$\{(1, (\hat{w}_{S_{i,1}})^2\cdot \hat{\boldw}\vert_{S_i \setminus \{S_{i,1}\}})(\boldx\vert_{S_i})^\intercal \}_{i=1}^{t}$$ to the server (computed over~$\bC$).
    \item The server computes and outputs $$\frac{2}{\sqrt{3}} \cdot \Im(\sum_{i=1}^{t}(\hat{w}_{S_{i,1}})^{2}\cdot(1, (\hat{w}_{S_{i,1}})^{2}\cdot \hat{\boldw}\vert_{S_i \setminus \{S_{i,1}\}})(\boldx\vert_{S_i})^\intercal) = \boldw\boldx^\intercal,$$ where $\Im(\cdot)$ denotes the imaginary part of a complex number.   
\end{enumerate}
The last step works since the imaginary parts of the elements of $R_3$ are $0$ and $\pm \frac{\sqrt{3}}{2}$.

\begin{theorem}
    For any $t \in [n]$, the above $\{0, \pm 1\}$ protocol has publication cost $d = (n - t)\log 3$, server-privacy $I(W; Q) = (n - t)\log 3$, and user-privacy $\ell = 2t$.
\end{theorem}
The proof works similarly to that of Theorem \ref{theorem:rootofunityprotocolmerit}, with the exception that $\ell = 2t$ because the user needs to send the real and imaginary parts of each inner product. This means that the above protocol is not optimal in terms of the privacy tradeoff; the publication cost is optimal nonetheless.
\section{Omitted proofs}\label{section:omitted}

\begin{proof}[Proof of Lemma~\ref{lemma:hadamardclaims}.(ii)]
    We prove this by induction. Consider the base case of $m = 1$. The rightmost column of $H_2$ is the submatrix of size $2^m \times m$ whose rows are precisely all vectors in $\{\pm 1\}^m$. Now suppose the claim holds for $m = k \ge 1$, and we would like to show that it is also true for $m = k + 1$. It can be easily seen that in $H_{2^k}$, the rightmost $k$ columns constitute the submatrix containing all vectors in $\{\pm 1\}^k$ as rows. If we take these columns in $H_{2^{k+1}}$ and also take the column with $1$'s in the first half and $-1$'s in the second half, then these $k+1$ columns constitute the submatrix that contains all vectors in $\{\pm 1\}^{k+1}$ as rows.
\end{proof}
\begin{proof}[Proof of Lemma~\ref{lemma:hadamardclaims}.(iii)]
    We prove this by induction. Recall Sylvester's construction for a Hadamard matrix of order $2^m$, where $m$ is a positive integer ($H_1 = [1]$):
    \begin{align*}
        H_{2^m} &= \begin{bmatrix}
        H_{2^{m-1}} & H_{2^{m-1}} \\
        H_{2^{m-1}} & -H_{2^{m-1}}
    \end{bmatrix}.
    \end{align*}
    Consider the base case of $m = 1$, where the matrix under consideration is 
    \begin{align*}
        H_{2} &= \begin{bmatrix}
        1 & 1 \\
        1 & -1
    \end{bmatrix}.
    \end{align*} It can be readily verified that the element-wise product of any two rows of $H_2$ is a row in $H_2$. Now suppose the statement holds for $m = k \ge 1$. For $H_{2^{k+1}}$, we have the following three cases.
    \begin{enumerate}
        \item The element-wise product of any two rows from the top half of~$H_{2^{k+1}}$ is a row in the top half. This follows from the induction assumption that the element-wise product of any two rows of~$H_{2^{k}}$ is a row of~$H_{2^{k}}$.
        \item Similarly, the element-wise product of any two rows from the bottom half of~$H_{2^{k+1}}$ is also a row in the top half. 
        \item From the induction assumption, it also follows that the element-wise product of a row from the top half and a row from the bottom half of~$H_{2^{k+1}}$ is a row in the bottom half. 
    \end{enumerate}
    We have thus argued that~$H_{2^{k+1}}$ satisfies the required property. \qedhere
\end{proof} 

\begin{proof}[Proof of Lemma \ref{lemma:pmvectorsinsubspace}]
    Let $M \in \bR^{\ell \times n}$ be a matrix whose row~$\bR$-span is $S$. Since $M$ has a reduced row-echelon form, $S$ can be written as $$S = \{
    (\boldv, L(\boldv))\vert \boldv \in \bR^\ell\}$$ for some linear transform $L: \bR^\ell \to \bR^{n-\ell}$, up to a permutation of entries. It follows that
    \begin{align*}
        |S \cap A^n| = \{
    (\boldv, L(\boldv))\vert \boldv \in A^\ell, L(\boldv) \in A^{n-\ell}\},
    \end{align*}
    which indicates that 
    \begin{align*}
        |S \cap A^n| &\le |\{
    (\boldv, L(\boldv))\vert \boldv \in A^\ell\}| = |A|^\ell.\qedhere
    \end{align*}
\end{proof}
\end{document}